\documentclass[journal]{IEEEtran}

\usepackage{graphicx}
\usepackage{epstopdf}
\usepackage[cmex10]{amsmath}
\usepackage{amsfonts}
\usepackage[ruled,linesnumbered]{algorithm2e}
\usepackage{algorithmic}
\usepackage{array}
\usepackage{eqparbox}
\usepackage{color}
\usepackage{times}
\usepackage{float}
\usepackage{stfloats}
\usepackage{graphicx}
\usepackage{subfigure}
\usepackage{type1cm}
\usepackage{url}
\usepackage{cite}
\usepackage{enumerate}
\usepackage{lineno}
\usepackage{amsthm}
\usepackage{amssymb}
\usepackage{bbm}
\usepackage{bm}
\usepackage{multirow}
\usepackage{mathtools}

\usepackage[short]{optidef}
\allowdisplaybreaks 

\usepackage{enumerate}

\SetKwInput{KwInput}{Input}
\SetKwInput{KwOutput}{Output}

\DeclareMathOperator*{\argmin}{arg\,min}

\newtheorem{Proposition}{Proposition}

\begin{document}


\title{A Novel Two-Layer DAG-based Reactive Protocol for IoT Data Reliability in Metaverse}



\author{Changlin Yang, Ying Liu, Kwan-Wu Chin, Jiguang Wang, Huawei Huang, Zibin Zheng

\thanks{
This article has been accepted by  43rd IEEE International Conference on Distributed Computing Systems (ICDCS 2023). 

C. Yang, Y. Liu, H. Huang and Z. Zheng are with the School of Software Engineering, Sun Yat-sen University, Zhuhai, Guangdong, China (e-mails: \{yangchlin6, liuy2368, huanghw28, zhzibin\}@mail.sysu.edu.cn).

K. Chin is with the School of Electrical, Computer and Telecommunications Engineering, University of Wollongong, NSW, Australia (e-mail: kwanwu@uow.edu.au).

J. Wang is with the School of Computer Science, Zhongyuan University of Technology, Zhengzhou, Henan, China (e-mail: 2020107237@zut.edu.cn)

Corresponding author is: Ying Liu. 
} }

\maketitle

\begin{abstract}
Many applications, e.g., digital twins, rely on sensing data from Internet of Things (IoT) networks, which is used to infer event(s) and initiate actions to affect an environment. This gives rise to concerns relating to data integrity and provenance. One possible solution to address these concerns is to employ blockchain. However, blockchain has high resource requirements, thereby making it unsuitable for use on resource-constrained IoT devices. To this end, this paper proposes a novel approach, called two-layer directed acyclic graph (2LDAG), whereby IoT devices only store a digital fingerprint of data generated by their neighbors. Further, it proposes a novel proof-of-path (PoP) protocol that allows an operator or digital twin to verify data in an on-demand manner. The simulation results show 2LDAG has storage and communication cost that is respectively two and three orders of magnitude lower than traditional blockchain and also blockchains that use a DAG structure. Moreover, 2LDAG achieves consensus even when 49\% of nodes are malicious.


\end{abstract}
\maketitle

%
\section{Introduction}
\label{Intro}
%

%
The Metaverse is poised to transform the Internet by closely coupling the digital and physical world~\cite{Lee2021AllON}, whereby digital twins (DTs) are used to represent a physical object/entity such as a person or machine or an entire factory.  
To do that, DTs require data from devices operating in Internet of things (IoT) networks.  For example, a device may collect health data from a person, which is then used by his/her digital twin~\cite{coorey2021health}.  
%
%
Indeed, an IoT network forms a bridge between the physical world and the Metaverse.  This is illustrated in Fig.~\ref{New_Fig_1}, where devices operating at the perception layer gather and send information via gateways to the Metaverse layer.  The DTs then use collected data to model and simulate a physical entity/process.  Further, these DTs may then initiate smart contracts based on received data~\cite{9076112}. 
Data reliability is thus a key concern to any users or DTs.  This is especially critical when DTs are relied upon for decision making; e.g., an operator may conduct simulations using DTs to determine how resources are to be allocated in a factory or when machines are to undergo maintenance.
Further, users/DTs may wish to audit collected data, which requires data to be traceable, immutable and transparent~\cite{9076112}.

One possible solution is to employ blockchain, where DTs leverage features such as immutability, traceability and integrity~~\cite{DTJurdak}. Consequently, a number of works have applied blockchain for secure data and information exchanges~\cite{sharma2020towards}. An example work is~\cite{malik2019trustchain}, which uses blockchain to provide end-to-end data integrity in a supply chain.  
%
%

%
\begin{figure}[t]
\centering
\includegraphics[width= 1\linewidth ]{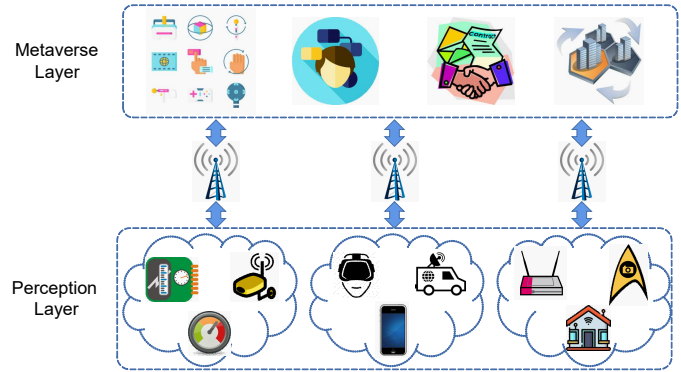}
\caption{An example Metaverse and IoT.  Devices collect data, which are then sent via gateways to Metaverse applications. \label{New_Fig_1}}
\end{figure}


Unfortunately, blockchain has a number of scalability issues that limit its use on IoT devices~\cite{zhou2020solutions}. 
%
Firstly, its storage requirement is significant, where each participating node needs to store a copy of the entire blockchain. For example, the size of the Bitcoin blockchain exceeds 380 GB in 2022 \cite{bitcoin_size} and the size of the Ripple blockchain, aka XRP ledger, is more than 14 TB \cite{XRP_ledger}. 
Secondly, its transaction throughput is limited, which relates to the speed in which data is written into a block.  For example, the average number of transactions processed by Bitcoin and Ethereum is respectively seven and 15 per second \cite{bach2018comparative}. Thus, it is not suitable for applications that generate frequent transactions.
%
Thirdly, its communication cost is significant, where an entire blockchain is sent to all participating nodes. For example, a Bitcoin node may upload more than 200 GB data per month~\cite{bitcoin_fullnode}. 
%
Lastly, we note that a key reason why blockchain has a high resource requirement is that it uses {\em proactive} consensus protocols. This means nodes spend resources to verify data via a consensus protocol even when a transaction or data is not used by participants.  Hence, all participants must store sufficient information to verify the data generated by other participants, which leads to high storage and communication overhead. 
To this end, we design a novel distributed approach, called a two-layer directed acyclic graph (2LDAG), to ensure data reliability in IoT networks.  
Briefly, each IoT node only stores its generated data and corresponding hash values, and those of its neighbors. 
Nodes only send and receive hashes to and from their neighbors, and they can generate proofs in parallel. The proofs form a DAG, linking all data in an IoT network. 
In addition, we propose a novel proof-of-path (PoP) protocol that allows any node to trace and verify the data of a source node. 
Advantageously, it is a {\em reactive} consensus protocol, meaning the data verification process is only initiated when a node needs to verify the data from a source node. 
In a nutshell, the contributions of this paper are as follows:
\begin{itemize}
    \item It proposes an architecture, called 2LDAG, that ensures the integrity and reliability of IoT data.  It embodies the following advantages over blockchain based solutions: improved storage, communication and throughput.
    \item It proposes a novel PoP protocol that operates in a reactive and distributed manner, and operates correctly even when a large number of malicious nodes exist. 
    \item It shows 2LDAG is secure against majority, Sybil, denial of service (DoS), eclipse and selfish attacks.
    %
    \item It presents the first study of 2LDAG.  
    Simulation results show that the storage and communication requirements of 2LDAG are two and three orders of magnitude lower than blockchain based solutions.
\end{itemize}

Next, Section~\ref{sec_related} discusses solutions that aim to scale blockchain and DAG blockchains. Section \ref{sec_arc} and \ref{sec_POP} demonstrate the 2LDAG architecture and outline PoP, respectively.  
The numerical simulation and results are shown in Section \ref{sec_sim}. Section \ref{sec_conclusion} concludes this paper. 

%
\section{Related Works\label{sec_related} }
The reliability of data in IoT or wireless sensor networks (WSNs) has been of interest for decades~\cite{sharma2020towards}. Traditional solutions leverage message authentication codes \cite{8354891} or public key infrastructures \cite{doukas2012enabling} to ensure data integrity. However, these solutions are centralized, which motivate the development of decentralized solutions such as blockchain~\cite{sharma2020towards}. 
%

%
%

To this end, we only review blockchain works that consider resource constrained nodes.  Specifically, we first review solutions that lower the storage and communication requirements of blockchain, and improve its throughput. We then discuss consensus protocols designed for IoT networks. 
\subsection{Blockchain Scalability}
Works that seek to scale blockchain have proposed light nodes, pruned nodes, sharding, coded blockchains and DAG~\cite{zhou2020solutions}. 
The following paragraphs review their main ideas and limitations. 

Nakamoto~\cite{nakamoto2008bitcoin} proposed so called light nodes that only store block headers. However, light nodes rely on full nodes to  store the entire blockchain and to verify data. 
Although light nodes have a lower storage requirement, Nakamoto's approach leads to data avaliability attacks~\cite{yu2020coded}. 
%

Pruned nodes remove information that is unnecessary for generating new blocks. 
For example, the nodes in~\cite{nakamoto2008bitcoin} remove unspent transaction outputs (UTXO), which is only utilized in one transaction and not used for block creation.  
%
However, both light and pruned nodes permanently delete data, meaning it is no longer possible to track the origin of blockchain data.
%

%
Sharding~\cite{wang2019sok} is a scheme whereby participants are divided into so called shards. Each shard maintains an independent sub-chain, and thus the participants in each shard have a comparatively lower storage requirement and higher throughput as compared to the conventional blockchain, especially with more shards. 
However, cross-shard activities require additional signaling complexity. Moreover, it sacrifices security due to the smaller number of participants in each shard. 
Coded blockchains~\cite{yang2022scaling} leverage error correction codes to encode blockchain data and store them at participants in a distributed manner.  
However, coded blockchains require extra computations for encoding and decoding. Moreover, there is a probability that an encoded block is not decodable~\cite{yang2022scaling}.

DAG blockchain~\cite{wang2020sok} changes the structure of blockchain from a chain to a graph. Consequently, transactions do not need to queue at the end of the chain, which significantly improves blockchain throughput. An example DAG blockchain is IOTA or Tangle~\cite{popov2018tangle}. Each transaction in IOTA verifies two previously generated transactions. There are no `miner' nodes where they generate transactions freely and simultaneously.
However, DAG blockchain nodes need to store the entire graph, i.e., the blockchain, to validate blocks generated by other nodes. 
Hence, it is not suitable for IoT devices with limited storage space. 

\subsection{Consensus Protocols}
Consensus protocols help nodes agree on information in a distributed manner even when some nodes are malicious or dishonest. 
For a comprehensive survey of consensus protocols, the reader is referred to~\cite{ferdous2020blockchain}.
Here, we will only discuss consensus protocols that are suitable for resource-constrained devices. They either use a proof based mechanism or practical byzantine fault tolerance (PBFT) protocols. 
%

%
%
Proof-of-stake (PoS)~\cite{salimitari2020survey} 
%
%
{\em miners} compete to generate blocks based on their {\em stakes} rather than computation power unlike PoW. 
Hence, PoS is suitable for IoT applications that use monetary rewards, which can be used as `stakes'.
In proof-of-elapsed-time (PoET)~\cite{bowman2021elapsed}, a node generates a block with a minimum random waiting timer. However, the verification in PoET requires a trusted execution environment. Such an environment is usually deployed on specific hardware, which is not suitable for a wide range of IoT applications~\cite{salimitari2020survey}.
%

%
PBFT based consensus protocols have high throughput, low latency and low computational overhead, which make them desirable for IoT networks~\cite{ferdous2020blockchain}. However, conventional PBFT has high communication overheads due to its three phases negotiation process. 
To this end, many works have seek to reduce the said communication overheads. 
%
For example, the work in~\cite{miller2016honey} proposed HoneybadgerBFT, which leverages error correction codes to reduce communication overhead when propagating blocks when using PBFT.  
They, however, incur extra computational complexity. 
%

\subsection{Discussion}
Existing blockchain scalability solutions have not jointly reduced storage, communication and computation cost whilst guaranteeing a sufficient level of throughput and security for IoT applications. A key reason is that these blockchains use proactive consensus protocols, where nodes must reach consensus on generated data. Inspired by the immutability, reliability and traceability features of blockchain, in the following section, we introduce 2LDAG and PoP.
\section{2LDAG Architecture \label{sec_arc} }
2LDAG has a {\em physical} layer and a {\em logical} DAG layer. At the {\em physical} layer, nodes maintain a list of neighbors, construct data blocks using generated data, and transmit/receive hashes to/from their neighbors.
The {\em logical} layer interconnects data blocks, whereby the hash of data blocks forms a DAG.  This DAG is used to verify a data block.

%
\subsection{Physical Layer} \label{phy_lay}
We model an IoT network as a graph $\mathcal{G}(\mathcal{V}, \mathcal{E})$.  It has a set $\mathcal{V}$ of static nodes inter-connected by an un-directed edge in the set $\mathcal{E}$. Each node is indexed by $i \in \mathcal{V}$, and at most $\gamma$ nodes are malicious.  Note, we do not consider trust management, where we assume there is a complementary method to register a device onto a network.
Each edge is denoted as $e_{ij} \in \mathcal{E}$, where $i, j\in \mathcal{V}$.  For node $i$, its neighbors set is
\begin{align}
    \mathcal{N}(i)  = \{j \;| \; e_{ij} \in \mathcal{E}, j \in \mathcal{V}\}.
\end{align}
All nodes have topological information, i.e., they know $\mathcal{G}(\mathcal{V}, \mathcal{E})$. 
%
%
Each node generates data at a rate of $r_i$ bit/s.  We emphasize that nodes only store their generated data; i.e., a node does not forward its data to any other nodes in $\mathcal{V}$. We use $m_i$ to denote the storage capacity of node $i$. 

%
Nodes store their generated data in so called {\em data blocks}, each of which has two segments: {\em body} and {\em header}. The body segment stores the data collected by a node and has a constant size of $C$ bits. This means node $i$ generates a data block every $C/r_i$ seconds.  The header segment contains data used to maintain the logical layer; see details in Section \ref{sec_data_block}. 

%
We use $b_{i, t}$ to indicate the $t$-th data block generated by node $i$; we also use $t$ to denote the generation time of $b_{i, t}$.  
Let $b^h_{i, t}$ and $b^d_{i, t}$ denote the header and body of $b_{i, t}$, respectively. 
Node $i$ records all its data blocks in the set $\mathcal{S}_i$.

After node $i$ generates data block $b_{i, t}$, it transmits the corresponding digest $H(b^h_{i, t})$ to all its neighbors in $\mathcal{N}(i)$, where $H(.)$ denotes a hash function. 
%
%
The data generation process at nodes will be presented in Section \ref{sec_block_generate}. 

\subsection{Data Block Structure}
\label{sec_data_block}
%
The data block structure used by nodes is shown in Fig.~\ref{fig_header}. A block header consists of the following fields: 
\begin{itemize}
    \item {\em  Version} is used to track changes and updates to the 2LDAG architecture and protocol.
    It has size $f^v$.
    
    \item {\em Time} records the generation time of a data block.
    We use $f^t$ to denote its size.

    \item {\em Root} corresponds to a Merkle tree root. Its size is $f^H$, which is also the hash size. 
    
    %
    \item {\em  Digests} record hashes that node $i$ has received from its neighbors, i.e., $H(b^h_{j, t}), \forall j \in \mathcal{N}(i) $, and the hash of its previous data block, i.e., $H(b^h_{j, t-1})$. The size of this field is $f^H \times (n+1)$, where $n = |\mathcal{N}(i)|$; note, $|.|$ denotes the cardinality of a set.
    

    \item {\em  Nonce} is a number that ensures the hash of version, time, root, digest and nonce meets a given difficulty level. This difficulty is chosen such that a device can find a suitable nonce quickly, see details in Section \ref{sec_block_generate}. Its size is $f^n$.  
    %

    \item {\em  Signature} is calculated using a public-private key scheme. 
    Each node $i$ has public key $pk_i$ and private key $sk_i$. Let $E(.)$ and $D(.)$ denote respectively an encryption and decryption function.   
    We can apply a low complexity encryption scheme; e.g., the one proposed in \cite{7784596} for smart home applications. 
    Let $m$ be the fields in a block header.  The signature field contains $s_{i, t} = E(m, sk_i)$. The size of $s_{i, t}$ is $f^s$. 

    
    %
    
\end{itemize}
\begin{figure}[t]
\includegraphics[width= 1 \linewidth ]{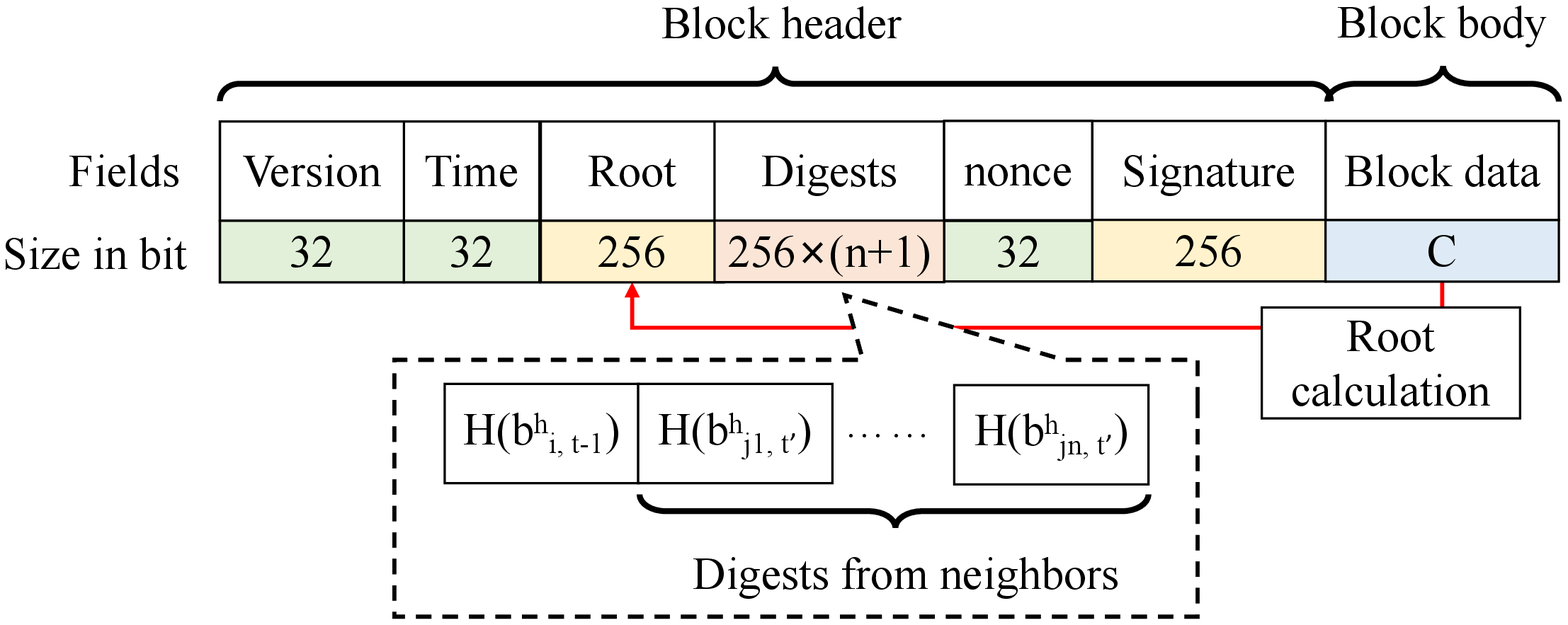}
\centering
\caption{Data block structure of 2LDAG.  The size of {\em Version}, {\em Time} and {\em Nonce} is 32 bits.  The {\em Root} and {\em Signature} has size 256 bits. The {\em Digests} field contains the digest of the latest block generated by a node, and the digests from its neighbors.  This means its size is $256 \times (n+1)$, where $n$ is the number of neighbors.  The body of all data blocks has size $C$ bits. }
\label{fig_header}
\end{figure}
Given that a block body has sampled data only, and has size $C$ bits, the total size of a data block at node $i$ is
\begin{align}
\label{equ_size}
    f_{i} = f_{c} + f^H (|\mathcal{N}(i)|+1)  + C,
\end{align}
where $f_{c}$ is a constant defined as 
\begin{align}
\label{equ_size_2}
    f_{c} = f^v + f^t +f^H + f^n + f^s.
\end{align}

%
\subsection{Logical Layer \label{sec_Logical_Layer}}
Define $\bar{\mathcal{G}}(\mathcal{B}, \mathcal{L})$ as a DAG, where the set of data blocks is $\mathcal{B} = \cup_{\forall i \in \mathcal{V}} \mathcal{S}_i$.  The set $\mathcal{L}$ contains directed edges that connect data blocks.  There is an edge $(b_{i, t}, b_{j, t'})$ in $\mathcal{L}$ if data block header $b^h_{j, t'}$ contains the digest of block header $b^h_{i, t}$.
%
%
We call $b_{i, t}$ a {\em parent} of $b_{j, t'}$, and $b_{j, t'}$ a {\em child} of $b_{i, t}$. In this paper, if the context is clear, we omit $i, t$ and write $b_{i, t}$ as $b_x$.
Define a {\em path} from data block $b_x$ to $b_y$ as
\begin{align}
\label{equ_p_x_y}
    \mathcal{P}_{x, y} &  = \{b_{n_1}, \ldots, b_{n_q}, \ldots ,b_{n_Q}\},
\end{align}
where $n_1 = x$, $n_Q = y$ and $Q$ is the path length. 
For any two adjacent data blocks $b_{n_q}$ and $b_{n_{q+1}}$ in path $\mathcal{P}_{h, k}$, data block $b_{n_q}$ is a {\em parent} of $b_{n_{q+1}}$.  
Further, we say data block $b_y$ is a  {\em descendant} of $b_x$.  
We then say node $i$ \textit{points} to block $b_x$ if data block $b_y$ in node $\mathcal{S}_i$ is a descendant of $b_x$.

\subsection{Data Block Generation \label{sec_block_generate} }
%
%
At network initialization, i.e., $t=0$, each node $i\in\mathcal{V}$ generates a genesis block; $b_{i, 0}$. Then, node $i$ calculates the digest $H(b^h_{i, 0})$ and transmits it to each neighbor node $j$ in $\mathcal{N}(i)$. 
Once node $i$ receives the digest $H(b^h_{j, 0})$ from neighbor $j$, node $i$ constructs a set $\mathcal{A}_i = \{H(b^h_{j, 0})\; |\; j \in \mathcal{N}(i)\}$,  which contain the latest digest from each neighbor in $\mathcal{N}(i)$; note, we have $|\mathcal{A}_i|= |\mathcal{N}(i)|$.
%
%

%
Given $C$ bits of data, node $i$ constructs data block $b_{i, t}$. 
In the header of $b_{i, t}$, node $i$ records the version and time, and calculates root $M(b^d_{i, t})$, where $M(.)$ is a Merkle tree root function.   It then includes in the {\em Digests} field, 
denoted by $\Delta_{i}$, the digests received from its neighbors as well as the digest of its previous data block $H(b^h_{i, t-1})$; i.e., we have $ \Delta_{i, j} = \mathcal{A}_i \cup \{H(b^h_{i, t-1})\}$. Next, node $i$ determines a nonce $n_{i, t}$ that satisfies
\begin{align}
    H(M(b^d_{i, t}),  \Delta_{i, j}, n_{i, t}) \le \rho, 
\end{align}
where $\rho$ is a small fixed difficulty level such that nodes can find a suitable $n_{i, t}$ quickly, e.g., in seconds. 
%
%
%
After that, node $i$ calculates the Signature field using Version $v$, Time $t$, Root $M(b^d_{i,t})$, Digests $\Delta_{i, j}$ and Nonce $n_{i, t}$. Formally, 
\begin{align}
\label{equ_s}
    s_{i,t} = E(H(v, t, M(b^d_{i,t}), \Delta_{i, j}, n_{i, t}), sk_i). 
\end{align}
%
%
Lastly, node $i$ transmits the digest $H(b^h_{i, t})$ to all its neighbors in $\mathcal{N}(i)$ and appends data block $b_{i, t}$ to its storage $\mathcal{S}_i$. After neighbor $j$ receives digest $H(b^h_{i, t})$ from node $i$, it updates the set $\mathcal{A}_j$ by replacing $H(b^h_{i, t-1}) \in \mathcal{A}_i$ with $H(b^h_{i, t})$. 
%

%

%



We now use an example to illustrate 2LDAG.  Referring to Fig.~\ref{fig_example1}, there are four physical nodes $A$, $B$, $C$, and $D$. Their neighbors set is respectively $\mathcal{N}(A) = \{B\}$, $\mathcal{N}(B) = \{A, C, D\}$, $\mathcal{N}(C) = \{B, D\}$ and $\mathcal{N}(D) = \{B, C\}$.  Assume node $D$ first generates a block, say $D_1$. 
It transmits digest $H(D_1^h)$ to node $B$ and $C$.  When node $C$ generates a block say $C_1$, it contains the digest $H(D^h_1)$ from node $D$. Node $C$ then transmits digest $H(C^h_1)$ to node $B$ and $D$. Similarly, digest $H(A^h_1)$ is sent to node $B$ after node $A$ generates block $A_1$. Hence, block $B_1$ generated by node $B$ includes digests $H(A^h_1), H(C^h_1)$ and $H(D^h_1)$. As shown in Fig.~\ref{fig_example1}, the corresponding digests stored in these blocks form a DAG. 
Observe that each node only stores data blocks generated by itself and transmits only block digests. For example, node $B$ only stores data block $B_1$ in its local storage, and the communication cost of node $B$ relates to the transmission and reception of three digests to and from nodes $A$, $C$, and $D$, respectively.  This explains why 2LDAG has a low storage and communication cost.
\begin{figure}[t]
\includegraphics[width= 1 \linewidth ]{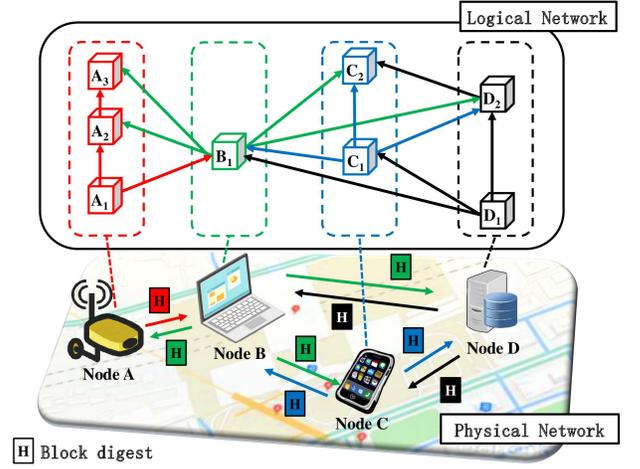}
\centering
\caption{A 2LDAG example. In the physical network, each node only stores its own data blocks; see colored dash boxes. Moreover, a node only transmits digests to neighbors; see colored squares with an `H'. In the logical network, data blocks are linked by directed edges and form a DAG.}
\label{fig_example1}
\end{figure}

\section{Proof-of-Path Protocol \label{sec_POP}}
%
%
%
A PoP process is initiated when a node, say $i$, needs to verify the content of a block stored at say node $j$, namely $b_x = b_{j,t} \in \mathcal{S}_j$.
We call node $i$ a {\em validator}, and node $j$ a {\em verifier}. 
Node $i$ first retrieves block $b_{j, t}$ from node $j$, which includes the block header $b^h_{j, t}$ and block body $b^d_{j, t}$. 
%
%
%
Using the said DAG, node $i$ aims to identify a set of nodes that are then stored in the set $\mathcal{R}_i$; these nodes contain a block that {\em points} to block $b_{j,t}$.  
Consensus of block $b_{j,t}$ is reached once the following condition is satisfied: $|\mathcal{R}_i| \ge  \gamma+1$, where $\gamma$ is the number of tolerable malicious nodes.
%
In other words, there is a path $\mathcal{P}_i$ containing at least $\gamma+1$ nodes that agree on the integrity of block $b_{i, j}$. 
%

%
%
%
%
Initially, we have $\mathcal{R}_i=\{j\}$ and $\mathcal{P}_i = \{b_{j, t}\}$.
Define the last added block on path $\mathcal{P}_i$ as a {\em verifying block}, denoted by $b_{v, t}$, which is initially set to $b_{v, t} = b_{j, t}$. 
To update the above sets, let $b_{j',t'}$ be a child data block of $b_{v, t}$.  The validator then retrieves the corresponding block header $b^h_{j',t'}$ and checks 
whether the value $H(b^h_{v,t})$ is in $b^h_{j',t'}$. 
%
If yes, the validator makes the following updates: (1) add block $b_{j',t'}$ to $\mathcal{P}_i$, (2) add node $j'$ to $\mathcal{R}_i$, and (3) replace $b_{v, t}$ with $b_{j',t'}$.
%

%
To illustrate PoP, consider Fig~\ref{fig_pop}.  Consider $\gamma = 2$, meaning consensus is reached if there are three nodes in $\mathcal{R}_i$.  Assume nodes $B$, $C$, and $D$ are neighbors of each other. Node $A$'s only neighbor is node $B$ and node $E$ only neighbor is node $D$. Assume there is a request to verify block $B1$. Then two possible paths are $\mathcal{P} = \{B1, D1, E2\}$ and $\mathcal{P}' = \{B1, A2, B2, C2\}$. Both paths contain three disjoint nodes, meaning there are $\gamma + 1$ nodes that directly or indirectly verify the integrity of block $B1$. Hence, consensus is reached if a validator constructs any of these two paths. However, path $\mathcal{P}'$ is longer, which means a validator that chooses path $\mathcal{P}'$ needs to retrieve an additional data block header. 
We next propose an algorithm to reduce the total number of block headers retrieved by a validator when selecting the next verifying node of data block $b_{v,t}$. 

\begin{figure}[htbp]
\includegraphics[width= 0.9\linewidth ]{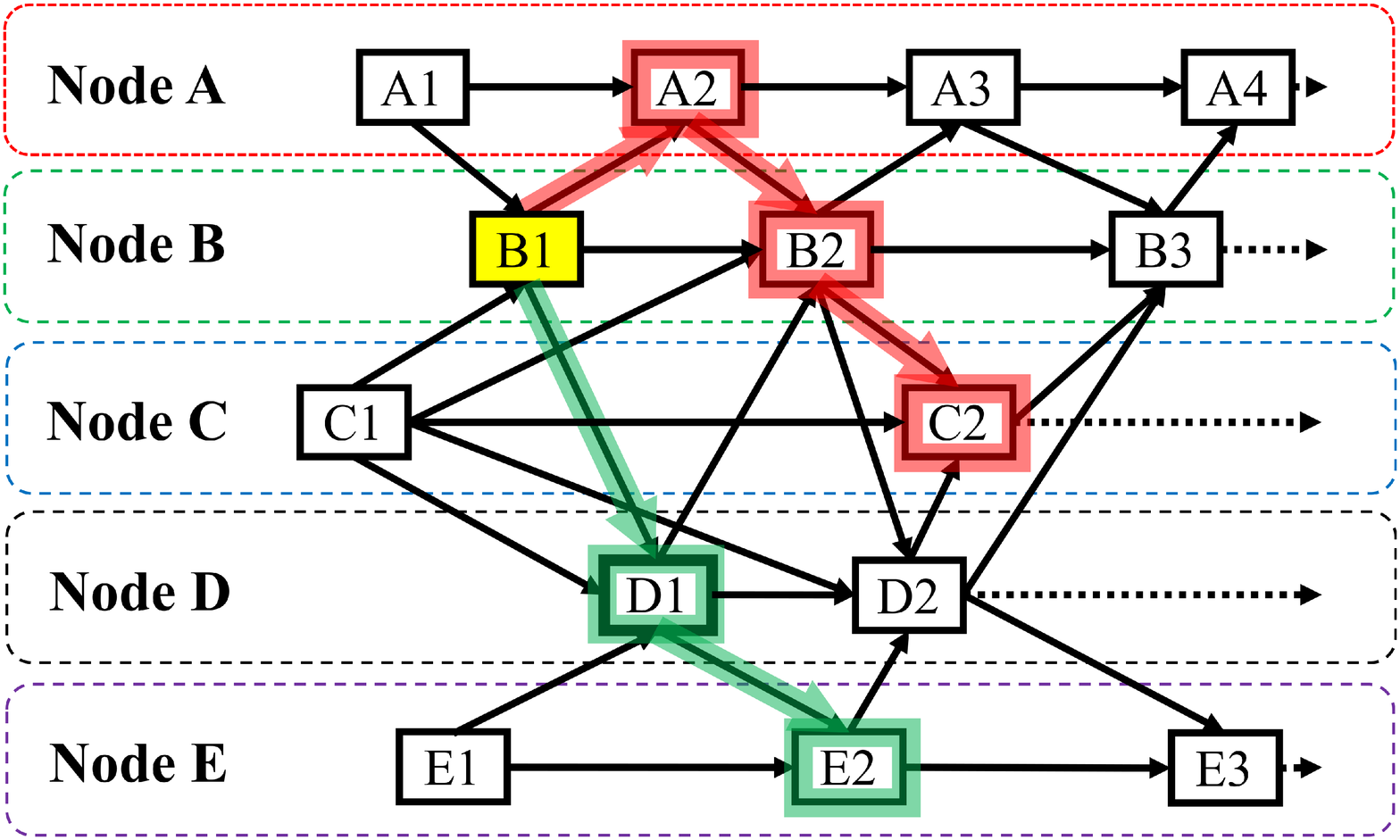}
\centering
\caption{An example of PoP. When validating block $B1$, the validator can obtain the header of block $D1$ and $E2$, see the green path or the header of block $A2$, $B2$ and $C2$, see the red path. Both paths achieve consensus, but the green path contains fewer number of data blocks. 
}
\label{fig_pop}
\end{figure}
%

\subsection{Weighted Path Selection (WPS) \label{sec_selec_neighbor}}
%
There are two cases to consider when selecting a verifying node $b_{v, t}$. In case-1, if a validator node $i$ selects a block stored in a node $j' \in \mathcal{N}(v)$ that has been included in $\mathcal{R}_i$ then $j'$ does not contribute to the consensus of block $b_{j, t}$.  Hence, the validator node $i$ needs to select a node $j'$ that is not in  $\mathcal{R}_i$. 
In case-2, all neighbors of node $v$ may be already in $\mathcal{R}_i$, then selecting any neighbor has the same effect.
To this end, we introduce a weight $w_v$ for node $v$. Here, $w_v$ is defined as the fraction of node $v$ neighbors that is not included in $\mathcal{R}_i$, i.e., 
\begin{align}
\label{equ_w_j}
    w_v = \frac{|\mathcal{R}_i \cap \{\mathcal{N}(v) \cup \{v\} |}{|\mathcal{N}(v)| +1 }.
\end{align}
%
%
%
Node $i$ then selects a $j' \in \mathcal{N}(v)$ with the minimum $w_{j'}$ value.  Mathematically,
\begin{align}
 j' = \argmin_{\hat{v} \in \mathcal{N}(v)} w_{\hat{v}}. 
\end{align}
If multiple neighbors have the same weight, the algorithm  selects the one not in $\mathcal{R}_i$. 

The pseudocode of WPS is shown in Algorithm~\ref{algo_path1}. 
Its inputs include $\mathcal{R}_i$, and $\mathcal{N}(v)$.
Its output is a node $j'$ that stores the child of $b_{v,t}$ to be included on path $\mathcal{P}_i$. 
%
%
Node $i$ first calculates weight $w_{j^*}$ using Eq.~\eqref{equ_w_j} for each $j^* \in \mathcal{N}(v)$; 
see lines \ref{algo1_line1}-\ref{algo1_line3}. It then selects the node with the minimum weight and adds it into a set of candidate nodes $\mathcal{Z}$; see line \ref{algo1_line4}. 
If there is only one node $j'\in \mathcal{Z}$, then it returns node $j'$; see lines \ref{algo1_line5}-\ref{algo1_line7}.
If we have $|\mathcal{Z}| \ge 2$, there are two cases: 1) if  $\mathcal{Z} \cap \mathcal{R}_i = \emptyset$ or $\mathcal{Z} == \mathcal{R}_i$, then all nodes in $\mathcal{Z}$ can be selected, or 2) if $\mathcal{Z} \cap \mathcal{R}_i \neq \emptyset $, then node $i$ selects one in $\mathcal{Z} \backslash \mathcal{R}_i$; see lines \ref{algo1_line8}-\ref{algo1_line13}.
\begin{algorithm}
    \KwInput{$\mathcal{R}_i$, $\mathcal{N}(v)$} 
    \KwOutput{$j'$} 
    \For{each $j^* \in \mathcal{N}(v)$ \label{algo1_line1} }{
    Calculate $w_{j^*}$ using Eq. \eqref{equ_w_j}
    }\label{algo1_line3}
    $\mathcal{Z} =\argmin_{\hat{v} \in \mathcal{N}(v)} w_{\hat{v}}$ \label{algo1_line4} \\
    \If{$|\mathcal{Z}| = 1$ \label{algo1_line5} }{
    Set $j'$ to the only element in $\mathcal{Z}$     
    }\label{algo1_line7}
    \If{$|\mathcal{Z}| \ge 2 \; \&\& \;  \mathcal{Z} \cap \mathcal{R}_i = \emptyset$ or $\mathcal{Z} == \mathcal{R}_i$ \label{algo1_line8} }{
    Set $j'$ to a random element in $\mathcal{Z}$
    }
    \If{$|\mathcal{Z}| \ge 2 \; \&\& \; \mathcal{Z} \cap \mathcal{R}_i \neq \emptyset $}{
    Set $j'$ to a random element in $\mathcal{Z} \backslash \mathcal{R}_i$
    } \label{algo1_line13}
    \label{algo_1_case_4}
    return $j'$
\caption{Weighted Path Selection (WPS)}
\label{algo_path1}
\end{algorithm} 
We now use Fig.~\ref{fig_pop} to illustrate WPS. When determining the next verifying block or node of $B1$, calculate the weight of node $B$'s three neighbors to yield $w_A = 0.5$, $w_C = 1/3$ and $w_D = 1/4$. Note that node $C$ does not have a block that is a child of $B1$, and node $D$ has the minimum weight. Hence, we choose $D1$ as the second block to be added onto the path at the validator. Next, when determining the verifying block or node of $D1$, we obtain the weight of its neighbors as $w_B = 0.5$, $w_C = 2/3$ and $w_E = 0.5$.  We choose $E2$ because node $B$ exists in $\mathcal{R}_i$.
\subsection{Trust Path Selection (TPS)\label{sec_TPS} }
From the example in Fig.~\ref{fig_pop}, we see that one may need to obtain $D1$ and $E2$ again when it verifies block $C1$.  This wastes both computation and communication resources. 
To this end, we require a node, say $i$, to store block headers that it has verified in a set $\mathcal{H}_i$. Then, node $i$ is able to directly construct a path $\{C1, D1, E2\}$ using the blocks in $\mathcal{H}_i$. 
This is carried out via a TPS algorithm, which is detailed next.
A validator node $i$ can verify block $b_{v, t}$ directly if its child block header is in the set $\mathcal{H}_i$, i.e., 
\begin{align}
    H(b^h_{v, t}) \in b^h_{j^*, t^*} \in \mathcal{H}_i,
\end{align}
where $b_{j^*, t^*}$ is a child block of the block $b_{v, t}$.
The validator node $i$ can then update the verifying block $b_{v, t}$ to $b_{j^*, t^*}$, add block $b_{j^*, t^*}$ to the set $\mathcal{R}_i$ and path $\mathcal{P}_i$. 
%
%
%

%
Algorithm~\ref{algo_length} details TPS.  
%
It iteratively updates $b_{v, t}$ using the data block headers in $\mathcal{H}_i$.  It returns the updated set $\mathcal{R}_i$, $\mathcal{P}_i$ and $b_{v, t}$ until no block headers in $\mathcal{H}_i$ contain $H(b_{v, t})$.
\begin{algorithm}
    \KwInput{$\mathcal{H}_i$, $\mathcal{R}_i$, $\mathcal{P}_i$, $b_{v, t}$} 
    \KwOutput{$\mathcal{R}_i$, $\mathcal{P}_i$, $b_{v,t}$} 
    \While{$\exists b^h_{j', t'} \in \mathcal{H}_i$ such that $H(b^h_{v, t}) \in b^h_{j', t'}$}{
    $\mathcal{R}_i = \mathcal{R}_i \cup \{j'\}$\\
    $\mathcal{P}_i = \mathcal{P}_i \cup \{b_{j', t'}\}$\\
    $b_{v, t} = b_{j', t'}$
    }
    return $\mathcal{R}_i$, $\mathcal{P}_i$, $b_{v,t}$ 
\caption{Trust path selection (TPS)}
\label{algo_length}
\end{algorithm} 
%

\subsection{Putting It All Together}
\label{sec_pop_detail}
We now discuss the procedures performed by a {\bf validator} node and a {\bf responder} node. Here, a responder node is defined as a node that sends information to a validator node in order to verify block $b_{v,t}$. 
%
%

{\bf Validator: }
The pseudocode of a validator is shown in Algorithm \ref{algo_validator}. 
%
%
Assume validator $i$ aims to verify block $b_{j, t}$.  It first retrieves $b_{j, t}$ from node $j$, which is also called the {\em verifier}. It then calculates the Merkle tree root $M(b^d_{j,t})$ and checks it against the header $b^h_{j, t}$. If there is inconsistency, i.e., $M(b^d_{j, t})$ does not equal $GetRoot(b^h_{j, t})$, PoP returns an error.   
Here, {\em GetRoot($b^h_{j,t}$)} returns the value in the {\em Root} field in the block header of $b^h_{j,t}$.
%
%
Otherwise, in line-3 Algorithm \ref{algo_validator}, validator $i$ makes the following settings: $\mathcal{R}_i=\{j\}$, $\mathcal{P}_i = \{b_{j, t}\}$ and the verifying block is set to $b_{v, t} = b_{j, t}$. 
Next, using the header $b^h_{j,t}$ of block $b_{j,t}$, validator $i$ calculates digest $H(b^h_{j,t})$ and constructs the set $\mathcal{R}_i$ and $\mathcal{P}_i$ using the set $\mathcal{H}_i$ as per TPS.
If consensus is not reached, i.e., $|\mathcal{R}_i| \leq \gamma$, the validator finds the next block for path $\mathcal{P}_i$ and adds a new node into the set $\mathcal{R}_i$.
%
Specifically, it calls WPS to find the next verifying node of data block $b_{v,t}$; WPS finds a node $j' \in \mathcal{N}(v)$ to verify block $b_{v, t}$.
%
%
After identifying node $j'$,  validator $i$ transmits a {\em REQ\_CHILD} message to node $j'$ containing $H(b^h_{v,t})$; see line-\ref{algo_3_16}.
%

%
Once validator $i$ receives a {\em RPY\_CHILD} message from node $j'$ within timeout $\tau$, see line-\ref{algo_3_18},
it extracts $b^h_{j', t^*}$ from the message; see line-\ref{algo_3_19}.
%
%
It then calculates $H(b^h_{j, t})$ and checks the digest stored in $b^h_{j', t^*}$. 
The function {\em GetDigest($b^h_{j', t^*}, v$)} returns the digest of node $v$ from block header $b^h_{j', t^*}$. 
%
%
If the above two digests are consistent, i.e., $H(b^h_{v, t}) == GetDigest(b^h_{j', t^*}, j)$, then node $j'$ {\em points} to block $b_{j, t}$.   Thus, validator $i$ adds node $j'$ into the set $\mathcal{R}_i$ and adds block index $b_{j', t^*}$ to the path $\mathcal{P}_i$. 
It then sets the verifying block $b_{v, t}$ to $b_{j', t^*}$.
Next, validator $i$ calls TPS to continue constructing the path; see line-\ref{algo_3_9}.
%
%
%
Otherwise, if the digests are not consistent or node $j'$ does not reply, node $i$ removes $j'$ from $\mathcal{N}(v)$ and selects another node to send a {\em REQ\_CHILD} message. 
If validator $i$ fails to receive a valid {\em RPY\_CHILD} message from all nodes in $\mathcal{N}'$, which is initialized to $\mathcal{N}(v)$ in line \ref{algo_3_line_13}, it rolls back to the previous verifying block and sends {\em REQ\_CHILD} message to nodes other than $v$; see lines \ref{algo_3_line_26}-\ref{algo_3_line_30}. 
To do this, it removes $v$ from $\mathcal{R}_i$, and also from $\mathcal{V}'$, which is initilized as $\mathcal{V}' = \mathcal{V}$ in line \ref{algo_3_line_14}.
Note that the set  $\mathcal{P}_i$ is a stake because the path in DAG has no loop. The validator then pops out $b_{v,t}$, which is the peak item, from $\mathcal{P}_i$ and sets verifying block to the block at the peak of updated $\mathcal{P}_i$. 
If the validator rolls back to the first block on the path, i.e., the verifier block $b_{j, t}$, and unables to receive a valid  {\em RPY\_CHILD} message, i.e., $\mathcal{R}_i == \emptyset$, Algorithm \ref{algo_validator} returns an error.
In contrast, if it receives a valid  {\em RPY\_CHILD} message, it sets $\mathcal{V}' = \mathcal{V}$ and continues to construct the set $\mathcal{R}_i$ and $\mathcal{P}_i$. 
%
%
%
%
%
%
%
%
Upon reaching consensus, i.e.,  $|\mathcal{R}_i| \ge \gamma + 1$, validator $i$ stores the block header $b^h_{j, t}$ for all blocks $b_{j, t}$ in $\mathcal{P}_i$ into the set $\mathcal{H}_i$; see line-\ref{algo_3_33}.  
%

%
%
\begin{algorithm}
    \KwInput{$b_{j, t}$, $\mathcal{G}(\mathcal{V}, \mathcal{E})$, $\mathcal{H}_i$} 
    \KwOutput{$\mathcal{H}_i$, `Success' or `Error' } 

    //*** Initialization ***// \\
    %
    $[b^h_{j, t}, b^d_{j, t}] == Request(b_{j, t})$ \\
    %
    \If{$M(b^d_{j, t})\; != \text{GetRoot}( b^h_{j, t})$}{return $\mathcal{H}_i$, `Error'}
    $\mathcal{R}_i=\{j\}$, $\mathcal{P}_i = \{b_{j, t}\}$ and $ b_{v, t} = \{b_{j, t}\}$ \\
    %
    //*** Construct path *** // \\
    \While{True \label{pop_while1} }{
    $(\mathcal{R}_i, \mathcal{P}_i, b_{v, t}) = \text{TPS}(\mathcal{H}_i, \mathcal{R}_i, \mathcal{P}_i, b_{v, t})$ \label{algo_3_9} \\ 
    %
    \If{$|\mathcal{R}_i| \ge \gamma + 1$}{
    break \\
    }
    $\mathcal{N}' = \mathcal{N}(v)$ \label{algo_3_line_13} \\ 
    $\mathcal{V}' = \mathcal{V}$ \label{algo_3_line_14}\\ 
    \While{$\mathcal{N}' \neq \emptyset$ \label{pop_while2} }{
    $j' = \text{WPS}(\mathcal{R}_i, \mathcal{N}')$ \\
    %
    Send(`REQ\_CHILD', $H(b^h_{v,t})$, $j'$) \label{algo_3_16} \\ 
    //*** Wait and receive message ***// \\
    \If{Msg = Receive(`RPY\_CHILD') \label{algo_3_18} }{
    %
    $b^h_{j', t^*} = GetHeader(Msg)$ \label{algo_3_19} \\ 
    %
    %
    \If{$H(b^h_{v, t}) == GetDigest(b^h_{j', t^*}, v)$}{ break}
    }
    $\mathcal{N}' \leftarrow \mathcal{N}'\;\backslash\; \{j'\}$ \\ 
    \label{pop_while2_end}
    \If{$\mathcal{N}' == \emptyset$ \label{algo_3_line_26} }{
    $\mathcal{R}_i.remove(v)$, $\mathcal{V}'.remove(v)$ \label{algo_3_line_pop1} \\
    $\mathcal{P}_i.pop(b_{v,t})$, \\
    %
    $b_{v,t} = \mathcal{P}_i.peak()$ \label{algo_3_line_resets}\\
    $\mathcal{N}(v)  = \{y \;| \; e_{vy} \in \mathcal{E}, y \in \mathcal{V}'\}$ \label{algo_3_line_resets_N}\\
        } \label{algo_3_line_30}
    \If{$\mathcal{R}_i == \emptyset$}{return $\mathcal{H}_i$, `Error'}
    }
    $\mathcal{R}_i \leftarrow \mathcal{R}_i \cup \{j'\}$, $\mathcal{P}_i \leftarrow \mathcal{P}_i \cup \{b_{j', t^*}\}$\\
    $b_{v, t} \leftarrow b_{j', t^*}$\\
    } \label{pop_while1_end}
    %
    %
    $\mathcal{H}_i \leftarrow \mathcal{H}_i \cup \{b^h_{v, t} \;|\; b_{v, t}\in \mathcal{P}_i\}$ \label{algo_3_33} \\ 
    return $\mathcal{H}_i$, `Success'
\caption{Validator}
\label{algo_validator}
\end{algorithm} 
{\bf Responder: }
%
The pseudocode of a responder is shown in Algorithm \ref{algo_responder}. 
It finds a block $b_{j', t^*}$ in $\mathcal{S}(j')$ whereby its block header $b^h_{j', t^*}$ contains $H(b^h_{v,t})$.
It then sends the validator a message containing $b^h_{j', t^*}$. 
Formally, node $j'$ finds a subset of child blocks $C_{j'}(b_{v,t})$ where 
\begin{align}
    C_{j'}(b_{v,t}) = \{b_{j', t'} \;|\;  H(b^h_{v,t}) \in b_{j', t'}, b_{j', t'} \in \mathcal{S}_{j'}\}. 
    \label{eq_C_j}
\end{align}
One or more blocks may contain the digest $H(b^h_{v,t})$, i.e., $|C_{j'}(b_{v,t})| > 1$. 
This is because nodes generate data blocks at different rates. If a node $v$ has a higher rate, its neighbor $j'$ may include the digest $H(b^h_{v,t})$ in multiple blocks. For example, as shown in Fig.~\ref{fig_example1}, the digest of $B_1$ is included in both $A_2$ and $A_3$. 
To this end, node $j'$ needs to select the oldest generated block $b_{j', t^*} \in C_{j'}(b_{v,t})$ where
\begin{align}
    t^* = \min \{ t' \;|\;  b_{j', t'} \in  C_{j'}(b_{v,t})\}.
    \label{eq_t_star}
\end{align}
%
%
After that, node $j'$ transmits a {\em RPY\_CHILD} message with  $b^h_{j', t^*}$ to  validator $i$. 
\begin{algorithm}
    \KwInput{$H(b^h_{v,t})$, $\mathcal{S}_{j'}$} 
    \KwOutput{ $b^h_{j', t^*}$} 
    //*** Listen for incoming connections ***// \\
    \If{Msg = Receive(`REQ\_CHILD') }{
    $H(b^h_{v,t}) = GetHash(Msg) $ \\
    
    %
    $C_{j'}(b_{v,t}) = \{b_{j', t'} \;|\;  H(b^h_{v,t}) \in b_{j', t'}, b_{j', t'} \in \mathcal{S}_{j'}\}$ \\
    Set $t^* = \min \{ t' \;|\;  b_{j', t'} \in  C_{j'}(b_{v,t})\}$ \\
    //*** Send message ***// \\
    Send(`RPY\_CHILD', $b^h_{j', t^*}$, $v$) \\
    }
\caption{Responder}
\label{algo_responder}
\end{algorithm} 

%
\subsection{Discussion}
Here, we provide some discussion on security concerns at the application layer.
For attacks/concerns at other layers of the protocol stack, we refer the reader to well-known attacks and defenses in the literature, e.g., \cite{8897627}.
%
%
We note that the registration and authentication of nodes are beyond the scope of this paper.  The reader is referred to works such as~\cite{WSNTrust} for authentication/security protocols.  
Lastly, we remind the reader that nodes have a public-private key pair, and they are aware of the topology and each other's public key.
%
\subsubsection{Malicious Nodes}
\label{sec_50}
Recall that in Section \ref{sec_pop_detail}, if validator $i$ does not receive or receive an invalid {\em RPY\_CHILD} message from $j' \in \mathcal{N}(j)$ within the timeout $\tau$, it requests a {\em RPY\_CHILD} message from  another node in $\mathcal{N}(j)$. 
If all nodes in $\mathcal{N}(j)$ fail to reply, validator $i$ creates a path that does not involve node $j'$; see Algorithm \ref{algo_validator}, lines \ref{algo_3_line_26}-\ref{algo_3_line_30}. 
Consider Fig.~\ref{fig_attack}.  Assume validator $E$ aims to verify a block at verifier $K$.  Malicious nodes are denoted by gray circles.  In this example, to ensure consensus, the resulting path $\mathcal{P}_E$ must contain at least seven nodes. 
After attempting to construct the green and blue path, the validator eventually constructs the red path.

\begin{figure}[t]
\includegraphics[width= 0.9\linewidth ]{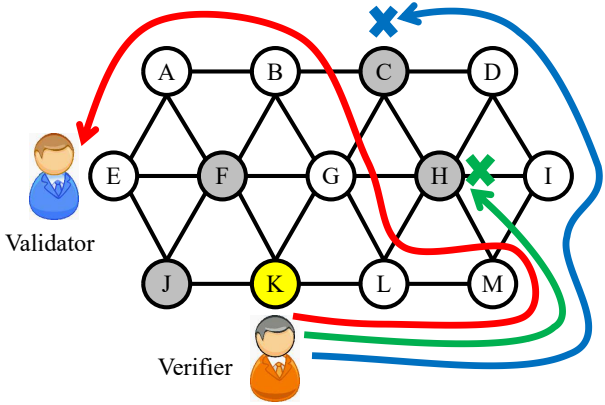}
\centering
 \caption{An example of PoP constructing a path around malicious nodes (gray circles). The yellow circle denotes the node with a block to be verified. The validator first constructs the green path, which terminates at malicious node $H$.  It then sends a request to $I$, which results in the blue path. Unfortunately, this blue path terminates at malicious node $C$.   As the validator has exhausted all neighbors of nodes $I$ and $D$, it then requests from the neighbor of node $L$, which results in the red path. }
\label{fig_attack}
\end{figure}

\subsubsection{Majority Attack}
\label{sec_51} 
In conventional blockchains, assuming proof-of-work (PoW), attackers with the majority of the computational power, i.e., 51\%, are able to corrupt blockchain content. 
%
This is because they are able to construct the longest chain containing invalid data, which is then accepted by participants.
In contrast, in 2LDAG, a node will never replace its blocks from those of other nodes.   
Specifically, nodes maintain their own blocks and they are not required to maintain a globally consistent 'chain'.
%

An attacker may cause all nodes in $\mathcal{R}_i$ or the path $\mathcal{P}_i$ to contain many malicious nodes. 
In this case, a solution is to increase the value of $\gamma$, meaning more nodes will be required to reach consensus.
%

\subsubsection{Sybil Attack}
A Sybil attacker creates fake nodes, which can then be used to launch a 51\% attack for example. 
Recall that a validator, say $i$, constructs a set $\mathcal{R}_i$ with unique nodes. This means a Sybil attacker that replicates the same malicious nodes does not have an impact on the consensus process.
Moreover, all nodes store the topology and they are able to use the public key of nodes to verify their identity.   Both facts prevent a Sybil attacker from adding nodes undetected. 


%
\subsubsection{Man-in-the-Middle Attack}
We only consider an attacker that aims to corrupt data block headers and/or body.   For an eavesdropping attack, a sender can encrypt its messages using its private key or derive a symmetric key based on its private key; see~\cite{SPINS} for an example.  An attacker could also drop messages between a source and destination pair.  To this end, we can employ multiple paths to bypass an attacker~\cite{SROUTING}. 

In 2LDAG, data block corruption is avoided as follows.  The data block header sent from a node, say $j$, is signed using its private key $sk_j$. A validator can validate this data block using the corresponding public key $pk_j$. This prevents manipulation by an intermediate node. 
Further, blocks are inter-connected in a DAG.  This means any change to a block will impact the digest of subsequent block(s) in the DAG.
%

%
\subsubsection{Denial of Service (DoS) Attack}
A malicious node may attempt to generate many digests and floods its neighbors.   In 2LDAG, a node needs to solve a puzzle that causes it to take seconds to generate a block.  Hence, a malicious is not able to generate a large number of blocks within a short time. This is the same strategy as IOTA \cite{popov2018tangle}.   
Further, a node may ban a neighbor that generates blocks quicker than the expected time to solve the puzzle. 
We note that such a malicious node will only affect its direct neighbors.  This is because the digests from a node are not flooded throughout a network.
%

%
%
%
In 2LDAG, nodes process the following messages: 1) {\em REQ\_CHILD}, 2) {\em RPY\_CHILD}, and 3) digest.
These messages and digests can include a nonce to avoid replay attacks.  Further, verifiers can authenticate a validator before replying to a REQ\_CHILD message.   Similarly, nodes only receive/forward/process digests, REQ\_CHILD and/or RPY\_CHILD messages from authenticated nodes.  Otherwise, they can safely discard these digests/messages.  This ensures the DAG and PoP involve only honest nodes.
%
%
%

%
Lastly, if a node is captured, hence becoming malicious, any changes to its data blocks will be detected by PoP.  This is because the digest of its data blocks will not be consistent with that of its neighbors.

\subsubsection{Selfish Attack\label{sec_selfish}}
In PoP, a selfish node may never reply to a {\em REQ\_CHILD} message. 
To prevent this, nodes can be equipped with a penalty mechanism.  Each node maintains a blacklist consisting of nodes that do not reply to a {\em REQ\_CHILD} message, either due to selfish behavior, disconnection or malicious. 
Then the block generated by the selfish node will not be verified by other nodes. The nodes in the blacklist will be removed after it helps transmit a certain number of blocks. Therefore, to show their willingness to participate after re-connection, nodes will actively transmit their data blocks.

    



%
\section{Performance Analysis}
In this section, we analyze the storage and message overhead of 2LDAG and PoP.  Briefly, our goal is to quantify the maximum amount of data stored by node $i$ at time $t$, and the number of messages emitted and received by a validator during verification.
\begin{Proposition}
\label{pop_1}
    The total number of data blocks at time $t$ is $\sum_{j \in \mathcal{V}} \bigl\lfloor \frac{t r_j}{C} \bigr\rfloor $. 
\end{Proposition}
\begin{proof}
    Recall that a node $j$ generates a data block every $\frac{C}{r_j}$ seconds. Hence, node $j$ has generated $\lceil \frac{t r_j}{C} \rceil$ data blocks at time $t$. Therefore, the total number of data blocks is $\sum_{i \in \mathcal{V}} \lfloor  \frac{t r_j}{C} \rfloor $.
\end{proof}
Recall that nodes, say $i$, maintain set $\mathcal{H}_i$ in order to reduce communication cost.  The following proposition bounds the size of $\mathcal{H}_i$. 
\begin{Proposition}
\label{pop_2}
    At time $t$, the size of $\mathcal{H}_i$ is upper bounded by $\frac{t(f_{c} + f^H |\mathcal{V}|)}{C} \sum_{j \in \mathcal{V}\backslash\{i\}} r_j$ bits. 
\end{Proposition}
\begin{proof}
    In the worst case, each node $i$ needs to store the header of all blocks in $\mathcal{H}_i$.  Let $f(\mathcal{H}_i)$ denote the size of $\mathcal{H}_i$, using Eq. \eqref{equ_size}-\eqref{equ_size_2} and Proposition-\ref{pop_1}, we have
    \begin{align}
        f(\mathcal{H}_i) & = \sum_{j \in \mathcal{V}\backslash\{i\}} \lfloor  \frac{t r_j}{C} \rfloor  \times (f_j - C) \nonumber \\
        & = \sum_{j \in \mathcal{V}\backslash\{i\}} \lfloor  \frac{t r_j}{C} \rfloor  \times (f_{c} + f^H (|\mathcal{N}(j)|+1)).
    \end{align}
    Note that we have $|\mathcal{N}(j)| \leq |\mathcal{V}|-1$ and $\lfloor  \frac{t r_j}{C} \rfloor \le \frac{t r_j}{C} $. Therefore, the above equation can be rewritten as 
    \begin{align}
        f(\mathcal{H}_i) & \le \sum_{j \in \mathcal{V}\backslash\{i\}}  \frac{t r_j}{C}   \times (f_{c} + f^H |\mathcal{V}|) \nonumber \\
        & = \frac{t(f_{c} + f^H |\mathcal{V}|)}{C} \sum_{j \in \mathcal{V}\backslash\{i\}} r_j. \label{eq_pop_2}
    \end{align}
    This ends the proof. 
\end{proof}

Hence, we have the following proposition for node storage.
\begin{Proposition}
    At time $t$, the total storage at node $i$ is upper bounded by $t r_i + \frac{t(f_{c} + f^H |\mathcal{V}|)}{C} \sum_{j \in \mathcal{V}} r_j$ bits. 
\end{Proposition}
\begin{proof}
    From Proposition-\ref{pop_1}, we see that the number of blocks stored by node $i$ is $|\mathcal{S}_i|  = \lfloor  \frac{t r_j}{C} \rfloor$. 
    %
    Therefore, the bit-size of $\mathcal{S}_i$, denoted by $f(\mathcal{S}_i)$, is 
    \begin{align}
        f(\mathcal{S}_i) & = \lfloor  \frac{t r_j}{C} \rfloor \times f_i \nonumber \\
        & = \lfloor  \frac{t r_i}{C} \rfloor \times ( f_c + f^H (|\mathcal{N}(j)|+1) + C) \nonumber \\ 
        & \le \frac{t r_i}{C} (f_{c} + f^H |\mathcal{V}| + C) \nonumber \\ 
        & = t r_i + \frac{t(f_{c} + f^H |\mathcal{V}|)}{C} \times r_i. \label{eq_pop_3_1}
    \end{align}
    Using Eq. \eqref{eq_pop_2} from Proposition-\ref{pop_2} and Eq. \eqref{eq_pop_3_1}, the total storage at node $i$ is bounded as follows:
    \begin{align}
        f(\mathcal{S}_i) + f(\mathcal{H}_i) & \le t r_i + \frac{t(f_{c} + f^H |\mathcal{V}|)}{C} \Big(r_i +  \sum_{j \in \mathcal{V}\backslash\{i\}} r_j\Big) \nonumber \\
        & = t r_i + \frac{t(f_{c} + f^H |\mathcal{V}|)}{C} \sum_{j \in \mathcal{V}} r_j.
    \end{align}
    This ends the proof.
\end{proof}

From Algorithm~\ref{algo_validator}, we analyze its {\em message overhead}, which is the  number of messages sent/received by a validator. We first have the following proposition that quantifies the message overhead lower bound.

%
\begin{Proposition}
    A validator node $i$ emits and receives at least $2(\gamma + 1)$ messages to reach consensus when $\mathcal{H}_i = \emptyset$.
\end{Proposition}
\begin{proof}
    When $\mathcal{H}_i = \emptyset$, the validator or node $i$ is not able to reach consensus using its stored information. This means it needs at least one {\em REQ\_CHILD} and {\em RPY\_CHILD} message exchange in order to verify block $b_{v, t}$.     
    Moreover, it needs to validate at least $\gamma + 1$ verifying blocks, and add disjoint nodes to $\mathcal{R}_i$, to reach consensus, i.e., $|\mathcal{R}_i| \ge \gamma + 1$. Hence, the validator emits and receives at least $2(\gamma + 1)$ messages to reach consensus when $\mathcal{H}_i = \emptyset$.  This completes the proof.
\end{proof}

\begin{figure}[t]
\includegraphics[width= 0.9\linewidth ]{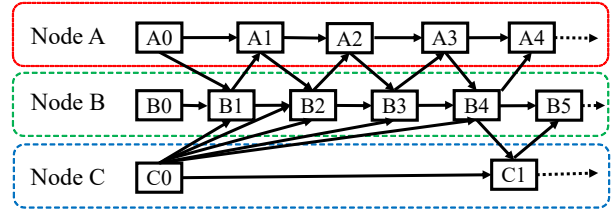}
\centering
\caption{An example of a {\em micro-loop} in PoP. When $r_B \gg r_c$ and $\gamma=2$.  PoP traverses the {\em micro-loop} $\{B2, A2, B3, A3, B4\}$ before reaching block $C1$. 
}
\label{fig_spiral}
\end{figure}

Before discussing the message overhead upper bound incurred by a validator, we first highlight an interesting finding. Consider Fig. \ref{fig_spiral}. 
Assume node $B$ has a much higher data generation rate than node $C$.  As a result, the digest of block $C0$ exists in multiple blocks at node $B$, and block $C1$ only has the digest of block $C0$ and $B4$.
%
Further,  assume $\gamma = 2$. When a validator needs to verify a block, say $B1$, it needs to construct a path $\{B1, A1, B2, A2, B3, A3, B4, C1\}$ that traverse all three nodes. 
Notice that the size of set $\mathcal{R}_i$ does not increase after adding the blocks in set $\{B2, A2, B3, A3, B4\}$ into the path.  This is because they involve nodes that exist in $\mathcal{R}_i$.
%
Note also there is a {\em micro-loop}, namely $\{B1, A1, B2, A2, B3, A3, B4\}$, in the path.
%
%
%
We use $\mathcal{M}$ to denote a set of nodes that are traversed by a micro-loop. 
The number of blocks in a micro-loop is limited by the block generation time interval of a node not in $\mathcal{M}$, say node $C$ in Fig. \ref{fig_spiral}. We then have the following proposition for the maximum number of blocks in $\mathcal{P}^l_i$. 
%
\begin{Proposition}
\label{pop_loop}
    Let $\mathcal{M}$ be a set of nodes that is traversed by a micro-loop, then the number of blocks within a micro-loop is upper bounded by $\sum_{i \in \mathcal{M}} \lfloor \frac{r_i}{\min \{r_j \;|\; j \in \mathcal{V}\backslash\mathcal{M}\}} \rfloor$. 
\end{Proposition}
\begin{proof}
%
%
%
%
Let $t'$ be the maximum block generation time interval of nodes in $ \mathcal{V}\backslash\mathcal{M}$, where
\begin{align}
    t' = \frac{C}{\min \{r_j \;|\; j \in \mathcal{V}\backslash\mathcal{M}\}}. 
\end{align}
Let $S_i(t)$ denote the number of blocks generated by node $i$ within time period $t$, where
\begin{align}
    S_i(t) = \lfloor \frac{t}{C/r_i} \rfloor = \lfloor \frac{r_i t}{C} \rfloor.
\end{align}
Therefore, the maximum number of blocks within a micro-loop is 
\begin{align}
    \sum_{i \in \mathcal{M}} S_i(t') = \sum_{i \in \mathcal{M}} \lfloor \frac{r_i}{\min \{r_j \;|\; j \in \mathcal{V}\backslash\mathcal{M}\}} \rfloor. 
\end{align}
\end{proof}

In terms of message overhead upper bound, we have the following proposition.
%
\begin{Proposition}
    Let nodes be indexed by the descending order of their data generation rate, i.e., $r_1 \ge r_2, \ldots, \ge r_{|\mathcal{V}|}$, and $ r_{\gamma} > r_{\gamma+1}$.  Assume there are no malicious nodes.
    The total  message overhead that a validator emits and receives is upper bounded by $(|\mathcal{V}(v)| + \gamma) (\sum_{j=1}^{j=\gamma} \frac{r_j}{r_{|\mathcal{V}|}} + \gamma+ 1)$.
    %
\end{Proposition}
\begin{proof}
    
    Recall that a validator node $i$ reaches consensus if $|\mathcal{R}_i|\ge \gamma+1$. 
    It needs to construct a path $\mathcal{P}_i$, which may include a micro-loop traversed by at most $\gamma$ nodes. Using Proposition-\ref{pop_loop}, the maximum number of blocks within a micro-loop is $\sum_{j=1}^{j=\gamma} \lfloor \frac{r_j}{r_{|\mathcal{V}|}} \rfloor - \gamma$. 
    Then the length of path $\mathcal{P}_i$ is upper bounded by 
    \begin{align}
        |\mathcal{P}_i| & \le \sum_{j=1}^{j=\gamma} \lfloor \frac{r_j}{r_{|\mathcal{V}|}} \rfloor + \gamma + 1.
    \end{align}
    When a validator validates block $b_{v, t}$ on path $\mathcal{P}_i$, it sends at most $|\mathcal{N}(v)|$ {\em REQ\_CHILD} messages, where $|\mathcal{N}(v)| \le |\mathcal{V}|-1$. 
    %
    In addition, it receives at most $\gamma+1$ {\em RPY\_CHILD} messages, where $\gamma$ of them are invalid as they are sent by malicious nodes. 
    Therefore, the total message overhead of a validator satisfies
    \begin{align}
        (|\mathcal{N}(v)|+ \gamma+1) |\mathcal{P}_i| & \le (|\mathcal{V}(v)| + \gamma) |\mathcal{P}_i| \nonumber \\
        & \le (|\mathcal{V}(v)| + \gamma) (\sum_{j=1}^{j=\gamma} \frac{r_j}{r_{|\mathcal{V}|}} + \gamma+ 1).
    \end{align}
    This ends the proof. 
    
    %
\end{proof}

\begin{figure*}[t]
\centering
\mbox{\hspace*{-1.5em}
\subfigure[C = 0.1 MB\label{storage_1000}]{\includegraphics[width= 0.25\linewidth ]{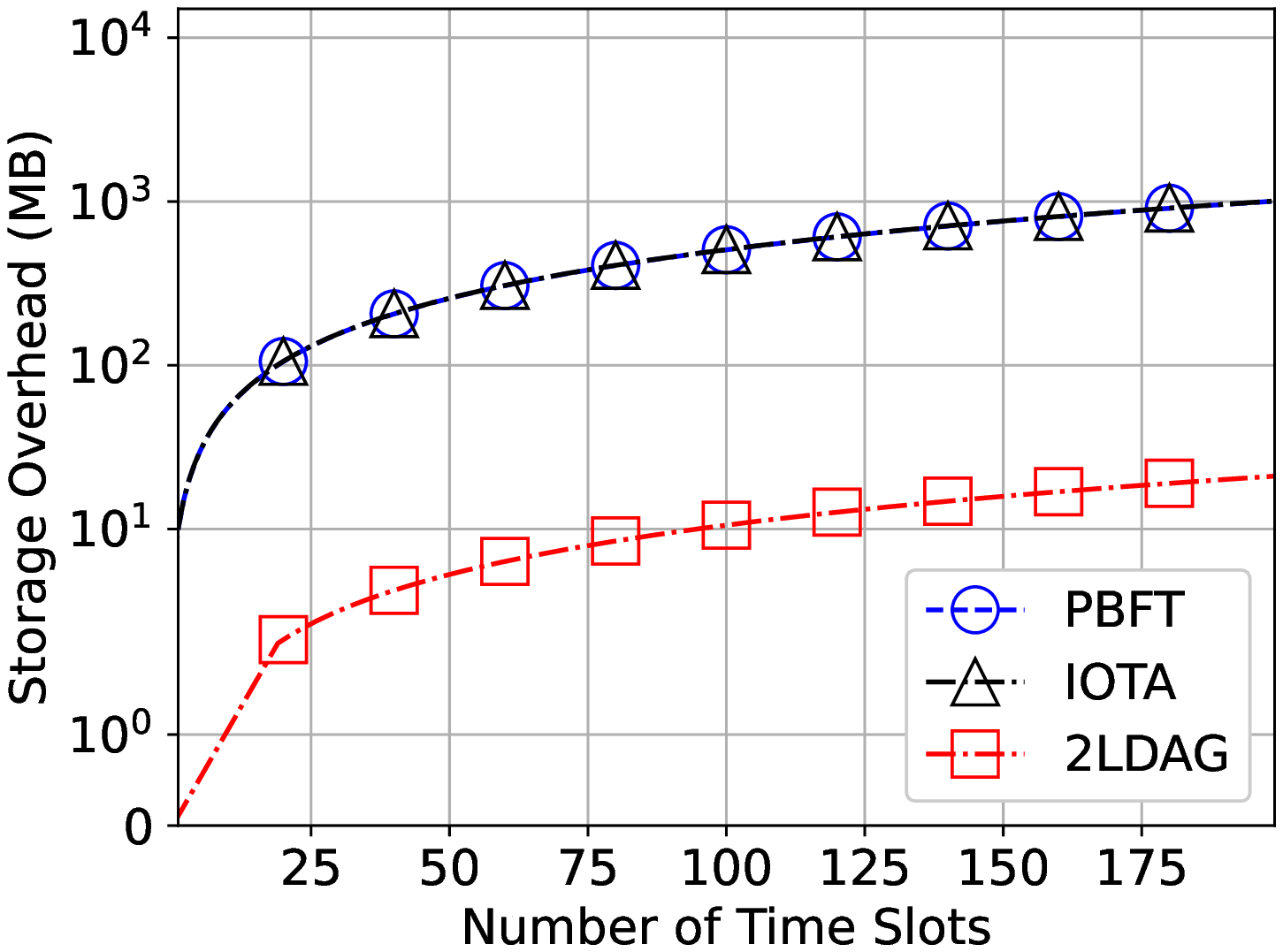}}\hspace*{-1.2em}
\subfigure[C = 0.5 MB\label{storage_3000}]{\includegraphics[width= 0.25\linewidth ]{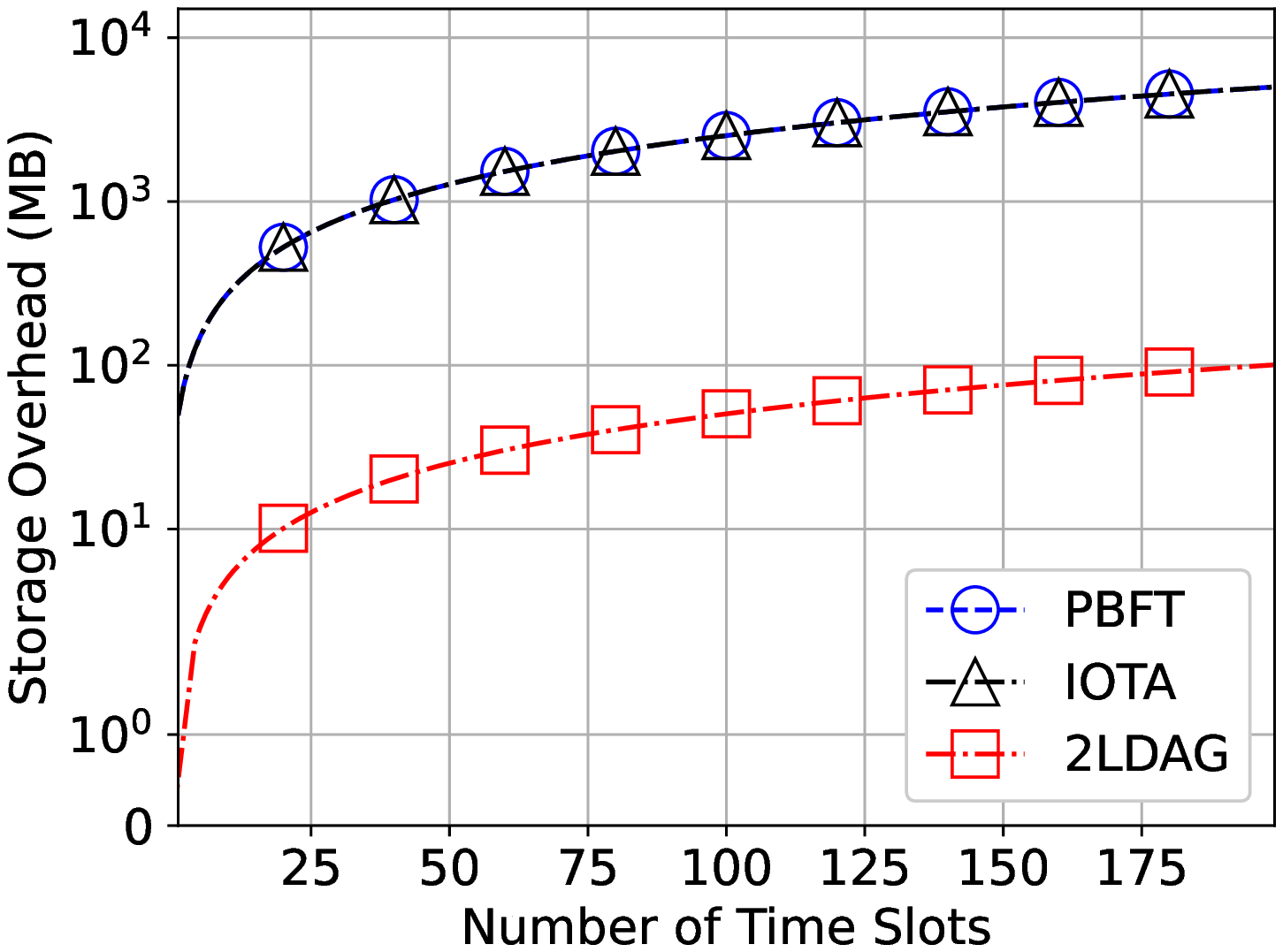}} \hspace*{-1.2em}
}
\subfigure[C = 1 MB\label{storage_5000}]{\includegraphics[width= 0.25\linewidth ]{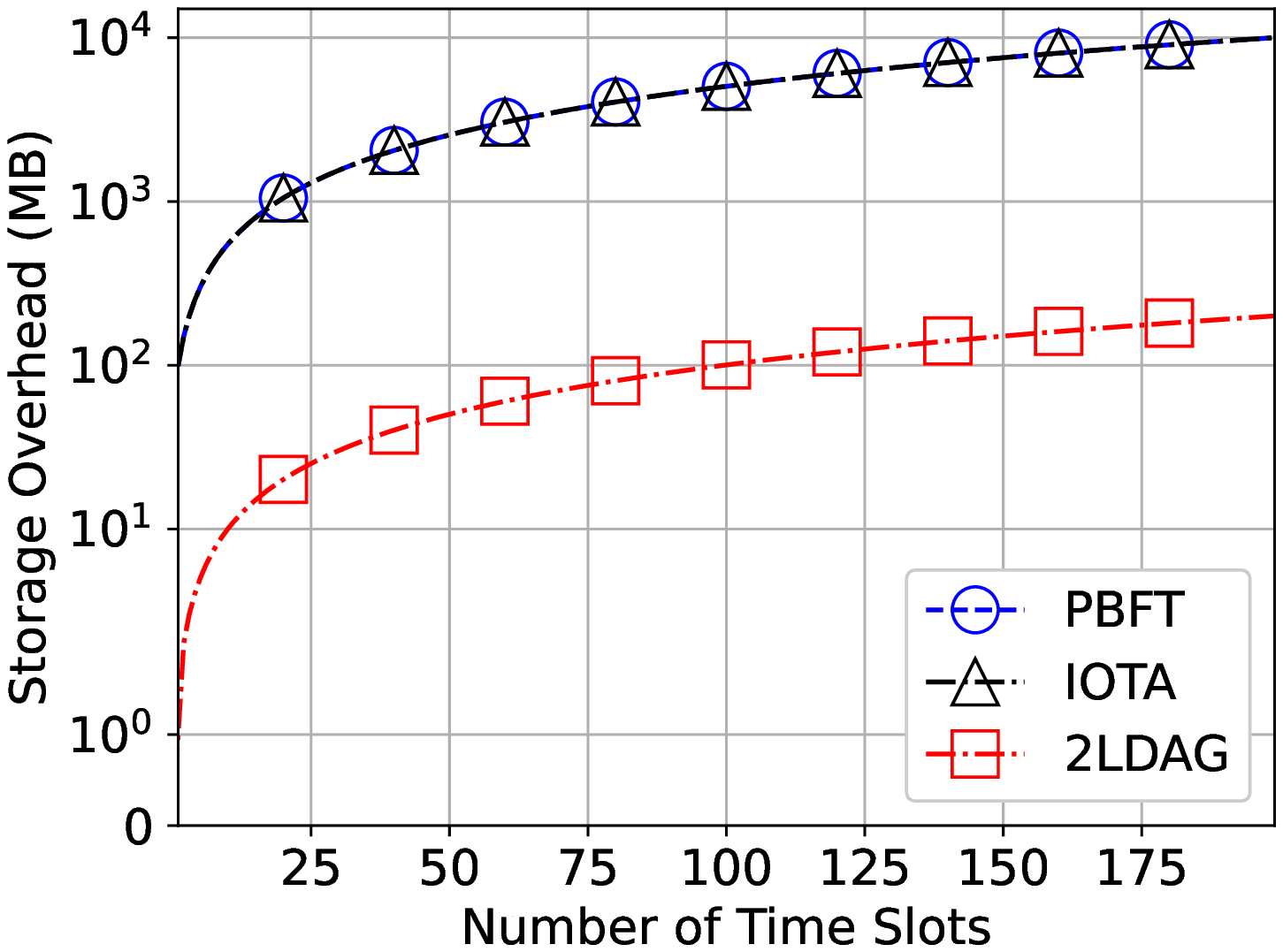}}\hspace*{-1.2em}
\subfigure[CDF of nodes storage when C = 0.5 MB \hspace*{-1.9em}\label{storage_5000_hist}]{\includegraphics[width= 0.25\linewidth ]{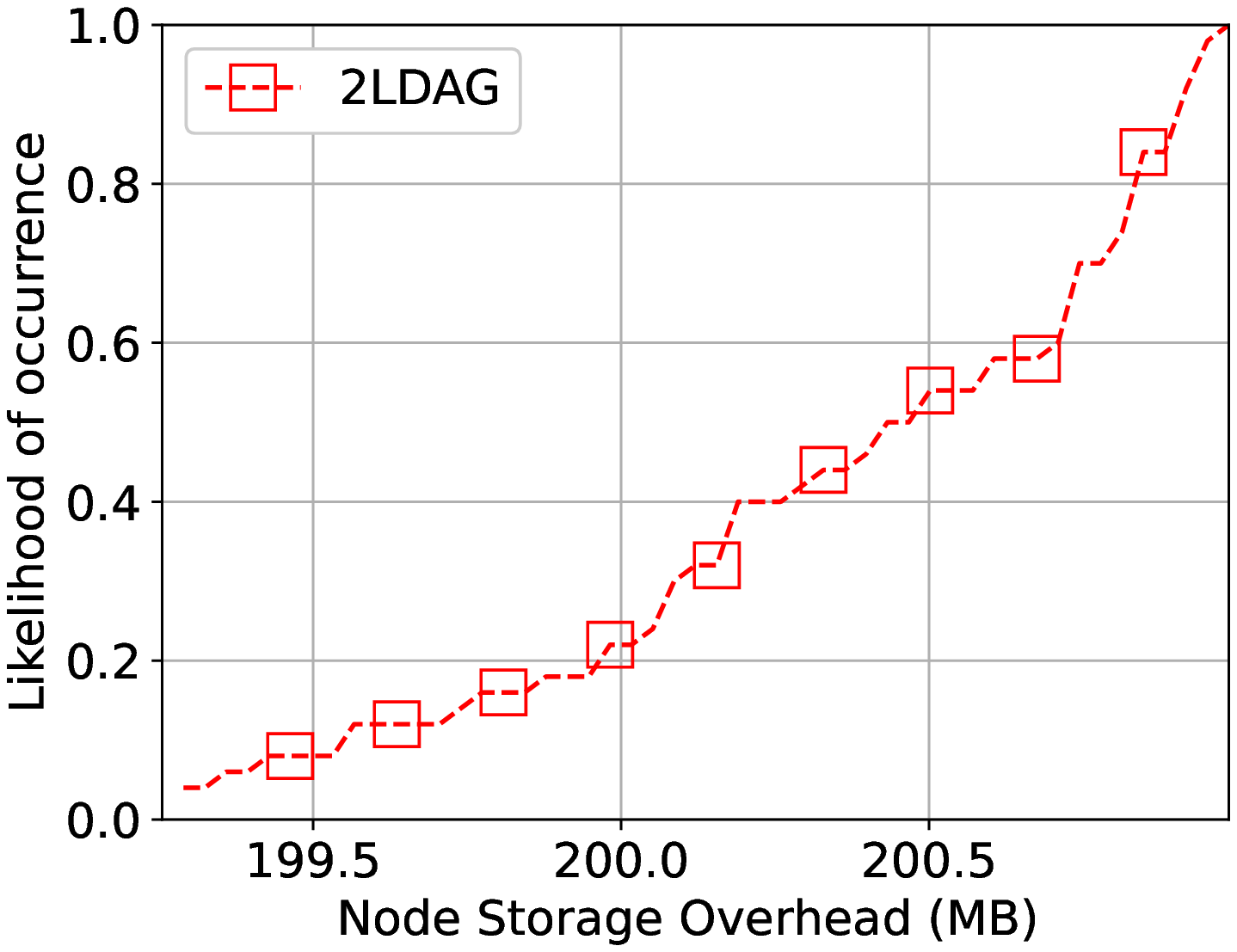}}
\caption{(a)-(c) The average node storage overhead (in log scale) of 2LDAG over PBFT and IOTA when the block body size is $C = \{0.1, 0.5, 1 \}$ MB, (d) the CDF of nodes storage overhead when block body size $C = 0.5$ MB at 200 time slots. All nodes generate one block at each time slot, i.e., $C/r_i=1, \forall i \in \mathcal{V}$. 
}
\label{storage_figs}
\end{figure*}

\section{Numerical Evaluation \label{sec_sim} }
The simulations are carried out on a desktop with an i7-12700 CPU and 32 GB RAM. 
We evaluate three key metrics: (1) {\em storage overhead}, which is the total disk space at a node, (2) {\em communication overhead}, which is the total amount of data a node transmits during block generation,  and (3) {\em time for consensus}, which is the time that PoP uses to find a sufficient path to reach consensus.
We compare the storage and communication overhead of 2LDAG with the PBFT blockchain\cite{castro1999practical}, and the tokenless IOTA blockchain \cite{popov2018tangle}. 
We set the size of Digests and Signature field respectively to $f^H = f^s =  256$ bits.  The Version, Time and Nonce field are set to  $f^v=f^t=f^n=32$ bits.  The physical network consists of 50 wireless IoT nodes in an area of 1000 square meters.
All nodes have a communication range of 50 meters.  To ensure a connected network, we place nodes one by one.  
That is, we start by randomly placing a node in the center of the said area.   A new node is then added to the area with the condition that it is always placed randomly within the communication range of an already deployed node.    

We divide time into slots.  Each node generates at most one block in each time slot. When a node generates a block, it must verify another block that is generated in the past using PoP.   This means a node is a validator when generating a data block.
%
It should be noted that PoP can only verify a block that is generated before $|\mathcal{V}|$  time slots.
This is because to achieve consensus for a block, a validator needs to find a directed path that traverses more than $\gamma  + 1$ physical nodes. This also means a path that passes more than $\gamma  + 1$ time slots. 
\subsection{Storage Overhead}
We first compare the average storage overhead of nodes in 2LDAG versus PBFT and IOTA.
%
We evaluate the following block sizes (in MB): $C=\{0.1, 0.5, 1\}$. 
From Fig.~\ref{storage_1000} to \ref{storage_5000}, we see that the node average storage overhead of 2LDAG is almost two orders of magnitude lower than PBFT and IOTA. 
Note, nodes have different block header sizes because they different number of neighbors. Hence, we show the Cumulative Distribution Function (CDF) of the storage at each node at 200 time slot in Fig.~\ref{storage_5000_hist}.
We see that the storage level at nodes varies from 199 to 201 MB. Therefore, the number of neighbors does not have a significant impact on node storage.


%
\begin{figure*}[t]
\centering
\mbox{\hspace*{-1.5em}
\subfigure[Overall DAG construction and consensus \label{comms_all}]{\includegraphics[width= 0.25\linewidth ]{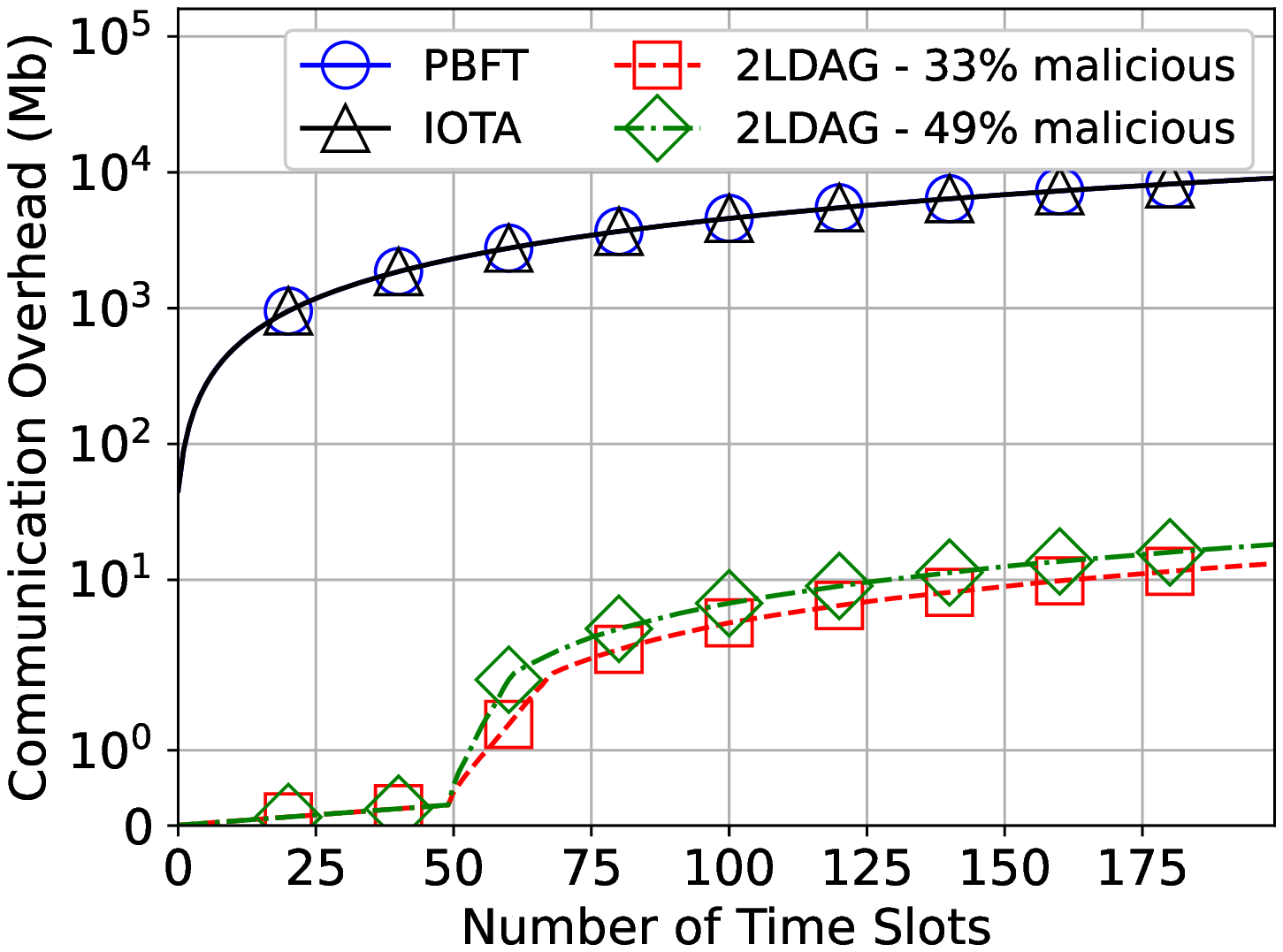}}\hspace*{-1.2em}
\subfigure[DAG construction \label{comms_generate}]{\includegraphics[width= 0.25\linewidth ]{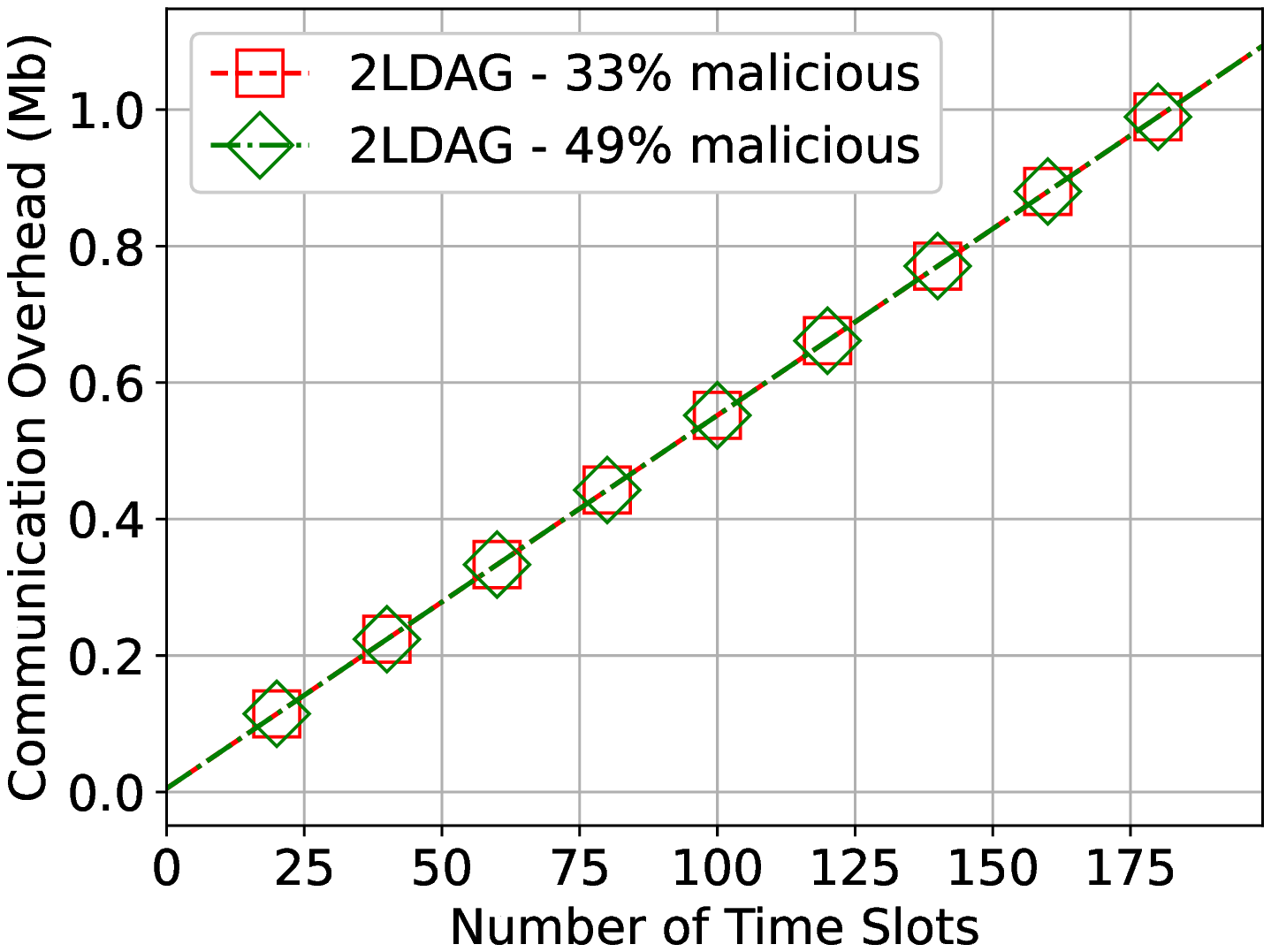}} \hspace*{-1.2em}
}
\subfigure[Consensus \label{comms_consensus}]{\includegraphics[width= 0.25\linewidth ]{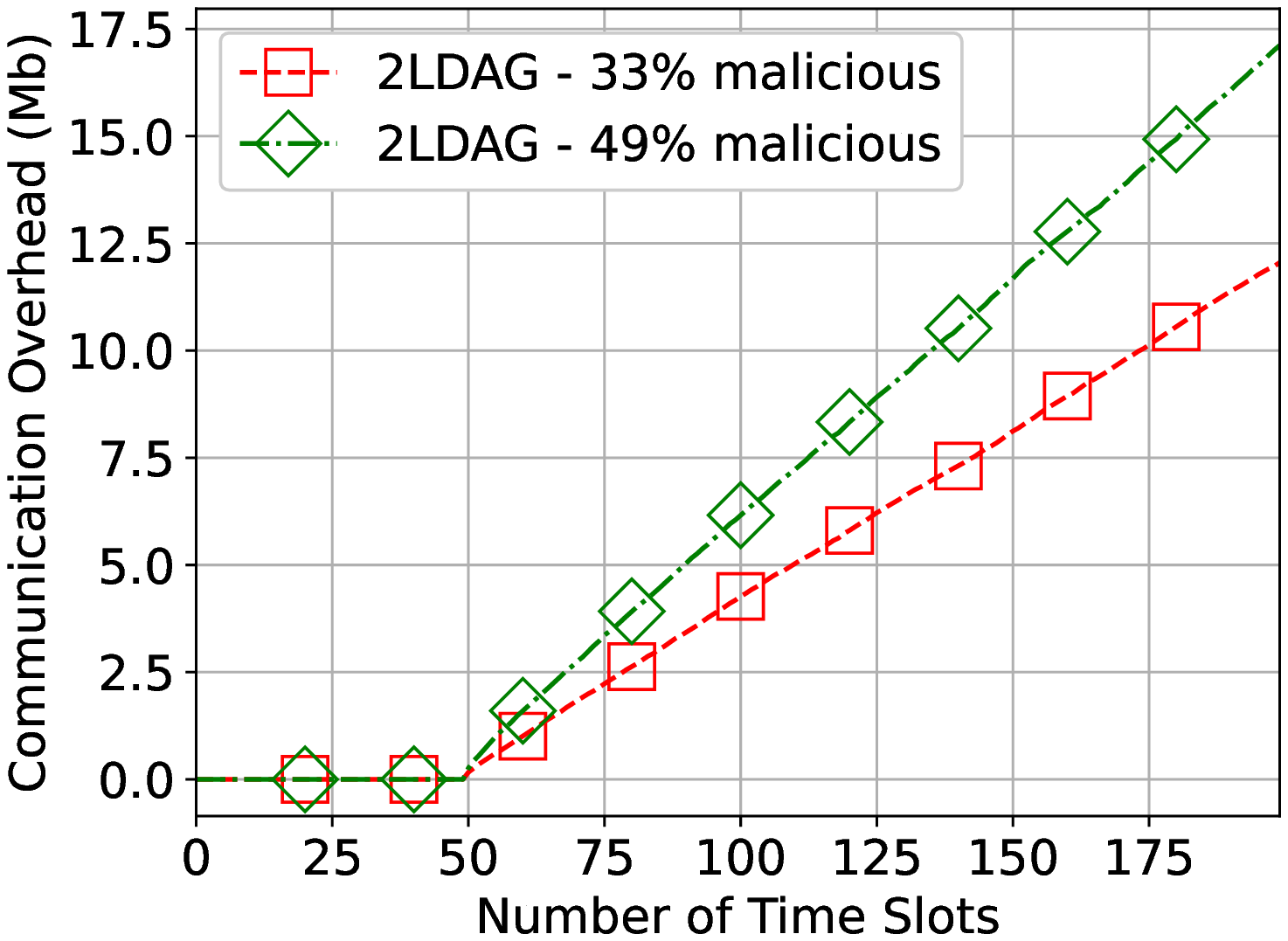}}\hspace*{-1.2em}
\subfigure[CDF of node overall communication overhead \label{comms_hist}]{\includegraphics[width= 0.25\linewidth ]{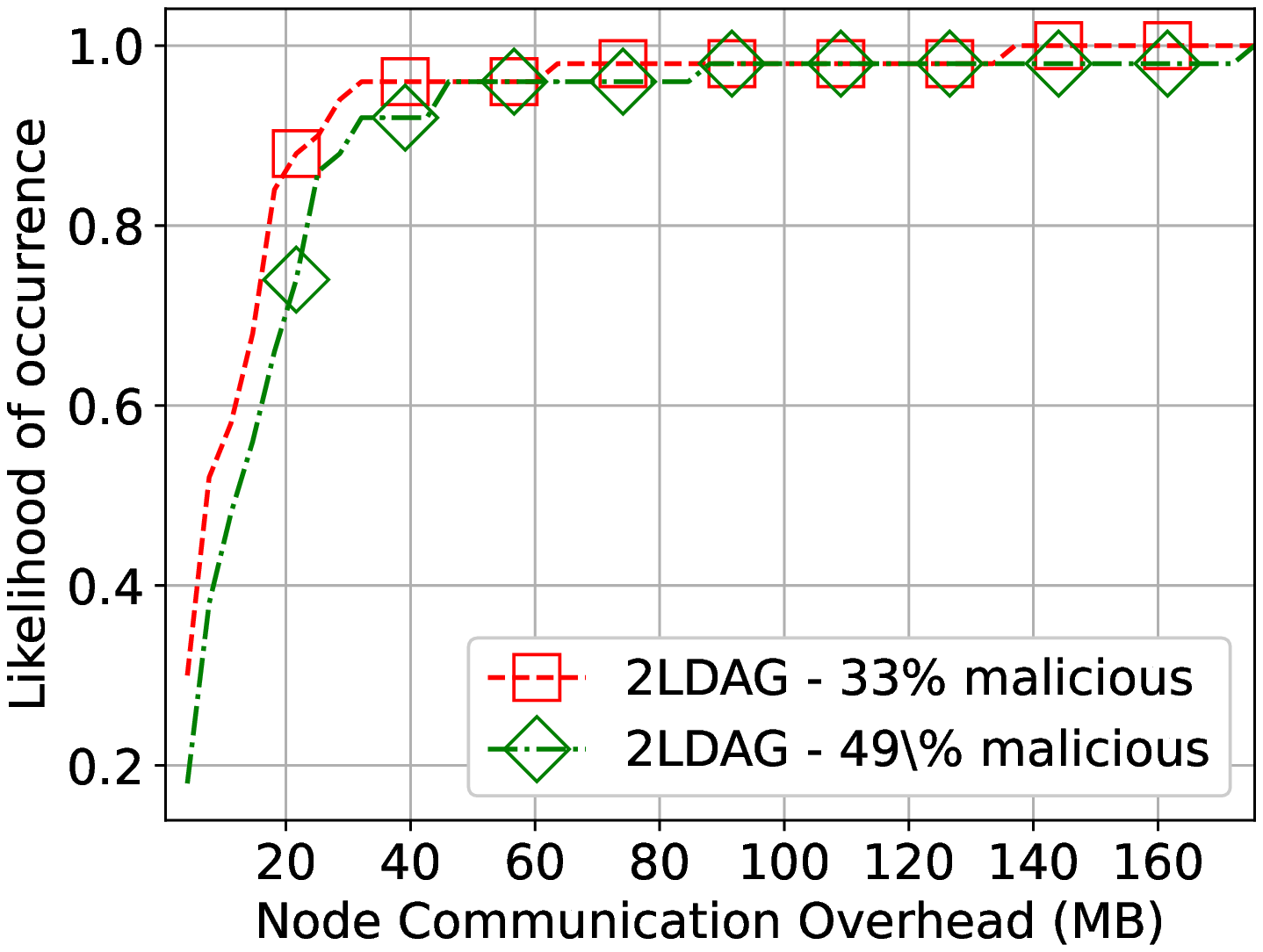}}
\caption{The average node communication overhead of (a) overall DAG construction and consensus, (b) DAG construction, (c) Consensus, and (d) the CDF of overall nodes communication overhead of 2LDAG when the tolerable malicious nodes $\gamma$ are 33\%$|\mathcal{V}|$ and 49\%$|\mathcal{V}|$, respectively. The block body size is $C=0.5$ MB. }
\label{Comms_figs}
\end{figure*}

\begin{figure*}[t]
\centering
\vspace*{-0.5cm}
\mbox{\hspace*{-1.5em}
\subfigure[$\gamma=10$ \label{Tolerable_10}]{\includegraphics[width= 0.25\linewidth ]{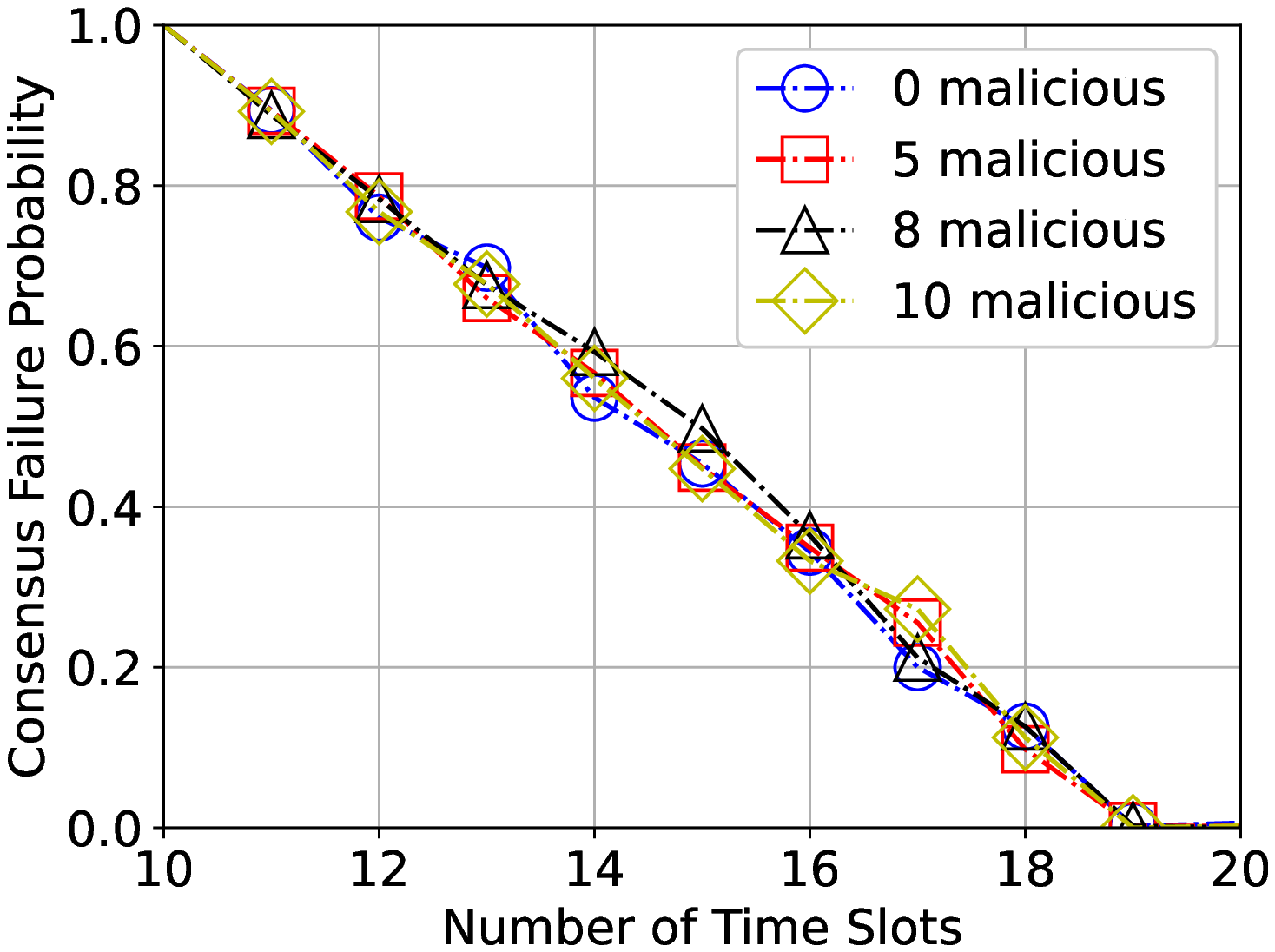}}\hspace*{-1.2em}
\subfigure[$\gamma=15$ \label{Tolerable_15}]{\includegraphics[width= 0.25\linewidth ]{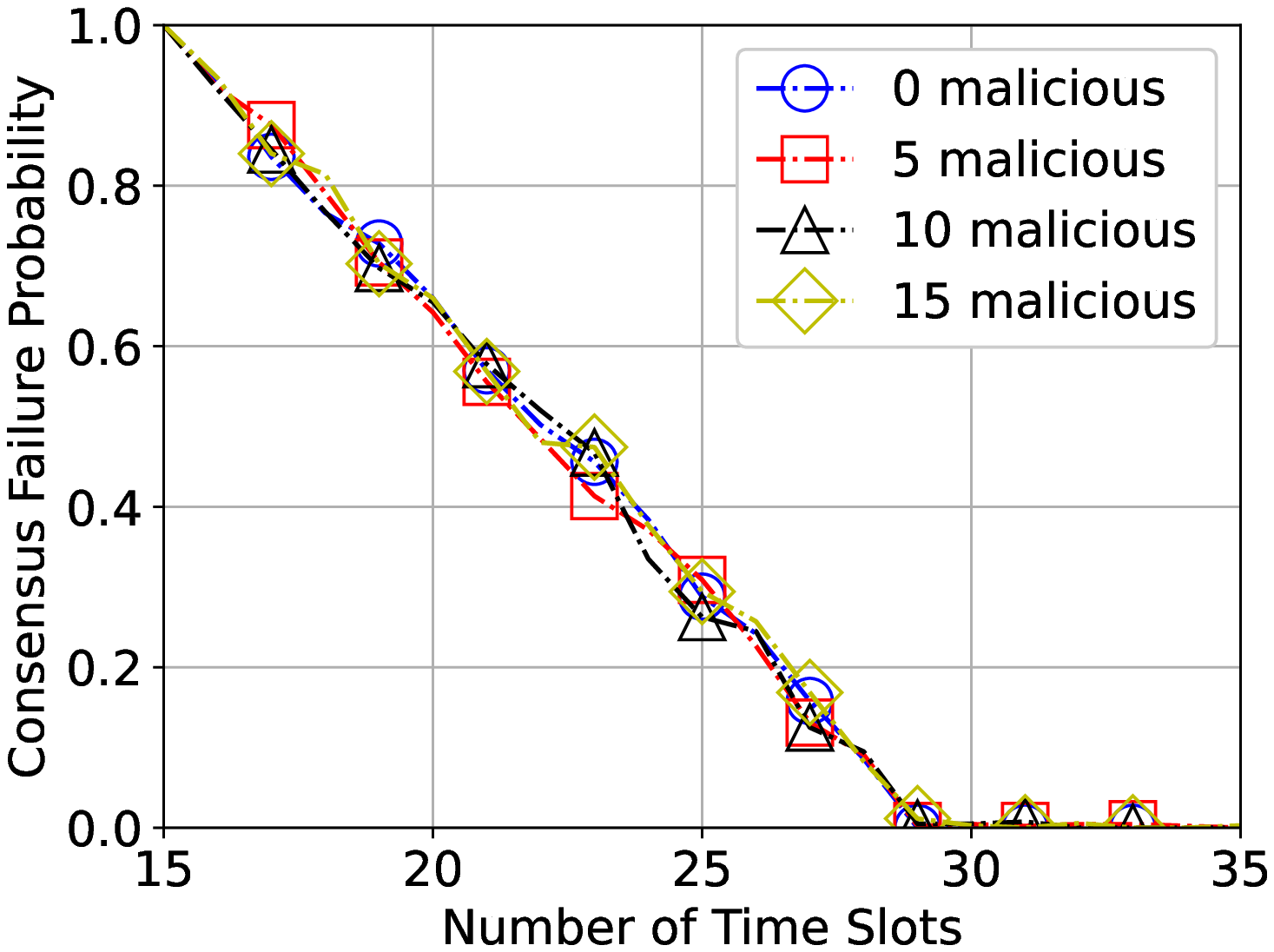}} \hspace*{-1.2em}
}
\subfigure[$\gamma=20$ \label{Tolerable_20}]{\includegraphics[width= 0.25\linewidth ]{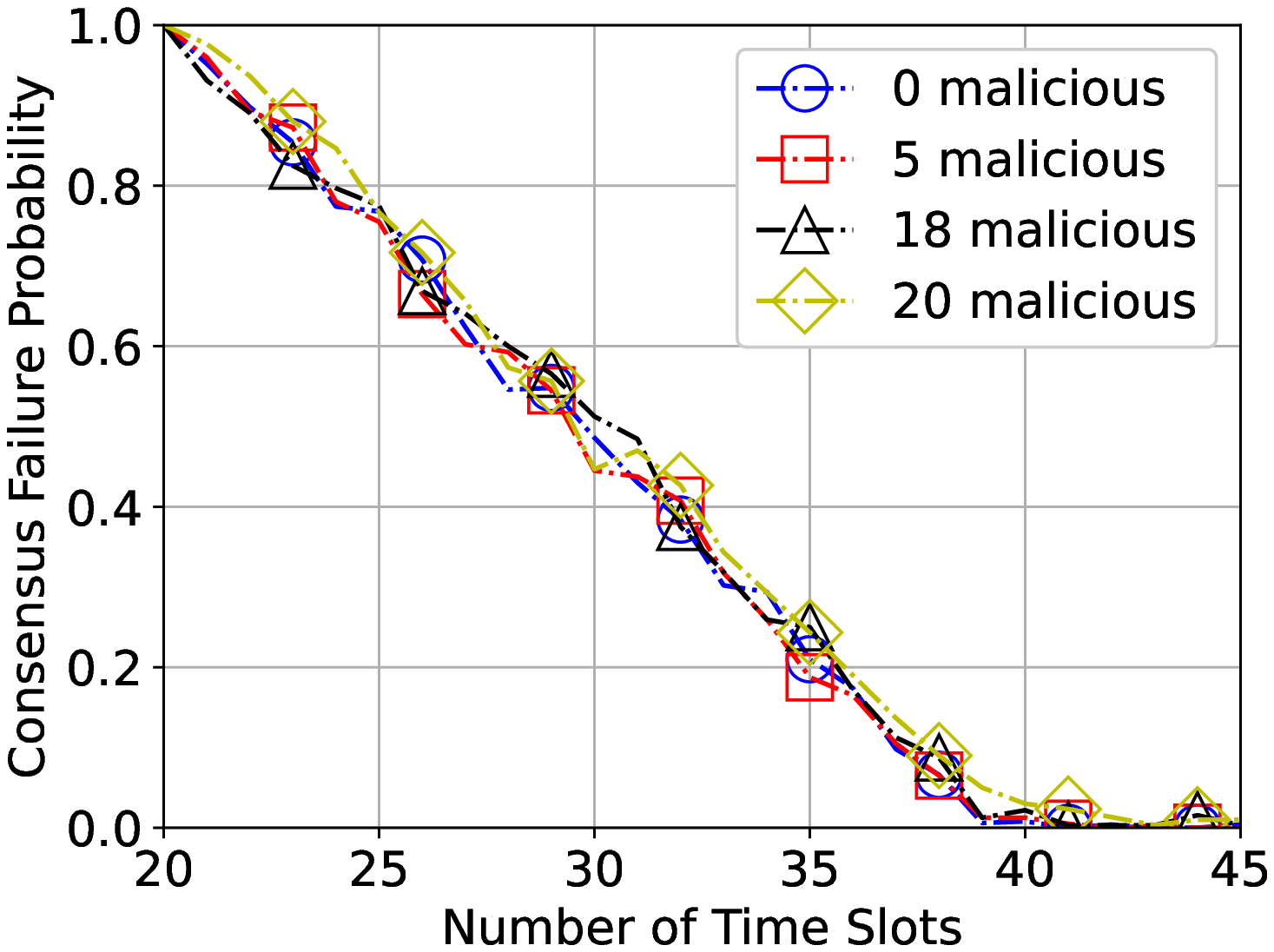}}\hspace*{-1.2em}
\subfigure[$\gamma=24$ \label{Tolerable_24}]{\includegraphics[width= 0.25\linewidth ]{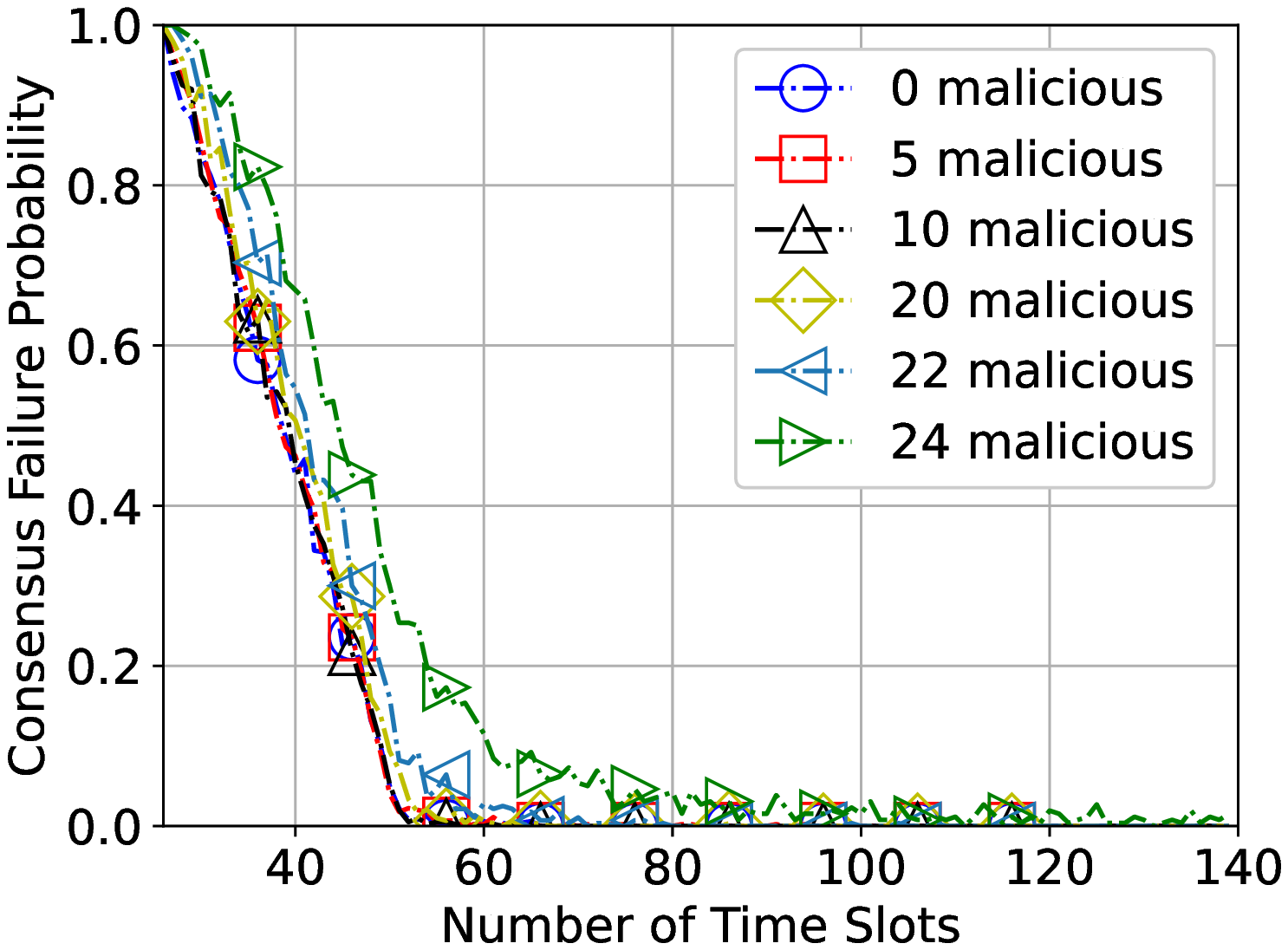}}
\caption{The number of time slots for 2LDAG to achieve consensus with malicious nodes occur when the number of maximum tolerable malicious nodes is $\gamma=\{10, 15, 20, 24\}$. The block body size is $C=0.5$ MB. Each node has a random block generation rate of one block per $\{1, 2\}$ time slots.
}
\label{attack_figs}
\end{figure*}





%
\subsection{Communication Overhead}
%
The block size is $C=0.5$ MB and each node generates one block at each time slot. We consider two possible number of malicious nodes, namely 33\% and 49\% of the total number of nodes are malicious. 
These values correspond to the number of tolerable malicious nodes in PBFT and IOTA or traditional blockchains such as Bitcoin.  This means 2LDAG reaches a consensus with paths that contain 17 and 26 physical nodes. From Fig.~\ref{comms_all}, we see that the overall node average communication overhead of 2LDAG is almost three orders of magnitude lower than PBFT and IOTA. Note that the communication overhead of 2LDAG is almost zero in the first 50 time slots. This is because nodes start to validate other blocks 
after $|\mathcal{V}|=50$ time slots.
Moreover, the communication overhead of 2LDAG for consensus is much higher than DAG construction, see Fig.~\ref{comms_generate} and \ref{comms_consensus}. 
This is because 2LDAG only transmits digests for block generation, but needs to transmit block headers for consensus. We see that 2LDAG which tolerates 49\% malicious nodes has a higher communication overhead for consensus, because it needs to construct longer paths to achieve PoP consensus. Fig.~\ref{comms_hist} shows the CDF of the communication overhead for each node at 200 time slot. We see that more than 90\% of the nodes transmit less than 40 MB of data, while others may transmit up to 160 MB of data. This is because a few nodes are important for forwarding data, which are vulnerable to attacks.

\subsection{Time for Consensus}
Lastly, we study the number of time slots required by PoP to reach consensus in 2LDAG when malicious nodes exist.  Each node generates one block per one or two time slots. We evaluate the consensus failure probability when 2LDAG verifies a block generated in the first $\gamma$ time slots. 
We say a consensus reaches when the consensus failure probability is zero.
We consider four scenarios where $\gamma\in\{10, 15, 20, 24\}$ in Fig.~\ref{attack_figs}. 
Note, for a network with 50 nodes, it can only tolerate up to 24 malicious nodes.
From Fig.~\ref{attack_figs}, the number of time slots to reach consensus increases with $\gamma$. This is because PoP requires a longer path. Moreover, the number of malicious nodes does not have a significant influence on the consensus time for the following $\gamma$ values: $\{ 10, 15, 20\}$.
However, consensus takes up to 120 time slots when $\gamma=24$. This is because a validator must find a path that contains all honest nodes to verify a block.
%


%
\section{Conclusions \label{sec_conclusion}}
This paper presents a novel 2LDAG architecture based on DAG and a PoP protocol to ensure IoT data integrity.
2LDAG has lower storage and communication costs as compared to solutions that leverage PBFT and DAG blockchains.    Hence, it is highly suited for use by IoT networks with resource constrained nodes. 
Moreover, 2LDAG has a high throughput. This is because all nodes are able to generate data blocks and transmit hashes independently.  Further,  the block generation rate of nodes is not limited by a consensus process. 
The simulation results show that 2LDAG reduces the storage cost by two orders of magnitude, and the communication cost by three orders of magnitude than PBFT and IOTA  blockchains.  
A potential future of 2LDAG is to construct the shortest path from a validator to a verifier in the physical layer when running the PoP protocol. This further reduces communication overhead for transmitting data block headers. Another future work is dynamic scenarios whereby nodes join and leave a network over time.


\bibliographystyle{ieeetr}
\bibliography{ref}
\end{document}